\documentclass[letterpaper, twocolumn]{IEEEtran}

\usepackage[mathscr]{eucal}
\usepackage[cmex10]{amsmath}
\usepackage{epsfig,epsf,psfrag}
\usepackage{amssymb,amsmath,amsthm,amsfonts,latexsym}
\usepackage{amsmath,graphicx,bm,xcolor,url}
\usepackage[caption=false]{subfig} 
\usepackage{fixltx2e}
\usepackage{array}
\usepackage{verbatim}
\usepackage{bm}
\usepackage{verbatim}
\usepackage{textcomp}
\usepackage{mathrsfs}
\usepackage{epstopdf}
\usepackage{bbm}

\catcode`~=11 \def\UrlSpecials{\do\~{\kern -.15em\lower .7ex\hbox{~}\kern .04em}} \catcode`~=13 

\allowdisplaybreaks[1]
 
\newcommand{\nn}{\nonumber}

\newcommand{\calA}{\mathcal{A}}

\newcommand{\calC}{\mathcal{C}}
\newcommand{\calD}{\mathcal{D}}

\newcommand{\calH}{\mathcal{H}}
\newcommand{\calI}{\mathcal{I}}

\newcommand{\calP}{\mathcal{P}}

\newcommand{\calS}{\mathcal{S}}

\newcommand{\calW}{\mathcal{W}}
\newcommand{\calX}{\mathcal{X}}
\newcommand{\calY}{\mathcal{Y}}
\newcommand{\calZ}{\mathcal{Z}}

\newcommand{\rmb}{\mathrm{b}}

\newcommand{\rmd}{\mathrm{d}}

\newcommand{\rme}{\mathrm{e}}

\newcommand{\rmJ}{\mathrm{J}}

\newcommand{\rmv}{\mathrm{v}}

\newcommand{\bbE}{\mathbb{E}}

\newcommand{\bbN}{\mathbb{N}}

\newcommand{\bbR}{\mathbb{R}}

\newcommand{\bbZ}{\mathbb{Z}}

\newcommand{\scF}{\mathscr{F}}

\DeclareMathAlphabet{\mathbsf}{OT1}{cmss}{bx}{n}
\DeclareMathAlphabet{\mathssf}{OT1}{cmss}{m}{sl}

\DeclareSymbolFont{bsfletters}{OT1}{cmss}{bx}{n}  
\DeclareSymbolFont{ssfletters}{OT1}{cmss}{m}{n}
\DeclareMathSymbol{\bsfGamma}{0}{bsfletters}{'000}
\DeclareMathSymbol{\ssfGamma}{0}{ssfletters}{'000}
\DeclareMathSymbol{\bsfDelta}{0}{bsfletters}{'001}
\DeclareMathSymbol{\ssfDelta}{0}{ssfletters}{'001}
\DeclareMathSymbol{\bsfTheta}{0}{bsfletters}{'002}
\DeclareMathSymbol{\ssfTheta}{0}{ssfletters}{'002}
\DeclareMathSymbol{\bsfLambda}{0}{bsfletters}{'003}
\DeclareMathSymbol{\ssfLambda}{0}{ssfletters}{'003}
\DeclareMathSymbol{\bsfXi}{0}{bsfletters}{'004}
\DeclareMathSymbol{\ssfXi}{0}{ssfletters}{'004}
\DeclareMathSymbol{\bsfPi}{0}{bsfletters}{'005}
\DeclareMathSymbol{\ssfPi}{0}{ssfletters}{'005}
\DeclareMathSymbol{\bsfSigma}{0}{bsfletters}{'006}
\DeclareMathSymbol{\ssfSigma}{0}{ssfletters}{'006}
\DeclareMathSymbol{\bsfUpsilon}{0}{bsfletters}{'007}
\DeclareMathSymbol{\ssfUpsilon}{0}{ssfletters}{'007}
\DeclareMathSymbol{\bsfPhi}{0}{bsfletters}{'010}
\DeclareMathSymbol{\ssfPhi}{0}{ssfletters}{'010}
\DeclareMathSymbol{\bsfPsi}{0}{bsfletters}{'011}
\DeclareMathSymbol{\ssfPsi}{0}{ssfletters}{'011}
\DeclareMathSymbol{\bsfOmega}{0}{bsfletters}{'012}
\DeclareMathSymbol{\ssfOmega}{0}{ssfletters}{'012}

\newcommand{\tilX}{\tilde{X}}

\newcommand{\hatz}{\hat{z}}

\newcommand{\tilz}{\tilde{z}}

\newcommand{\barz}{\bar{z}}

\newcommand{\eps}{\varepsilon}

\def\fndot{\, \cdot \,}

\DeclareMathOperator*{\argmin}{arg\,min}

\newcommand{\bone}{\mathbbm{1}}

\newtheorem{theorem}{Theorem} 
\newtheorem{lemma}{Lemma}

\newtheorem{definition}{Definition}

\newcommand{\markov}{\mathrel{\multimap}\joinrel\mathrel{-}%
\joinrel\mathrel{\mkern-6mu}\joinrel\mathrel{-}}

\newcommand{\qednew}{\nobreak \ifvmode \relax \else
      \ifdim\lastskip<1.5em \hskip-\lastskip
      \hskip1.5em plus0em minus0.5em \fi \nobreak
      \vrule height0.75em width0.5em depth0.25em\fi}

\usepackage{cite}

\begin{document}
\flushbottom
\allowdisplaybreaks[1]

\title{Minimum Rates of Approximate Sufficient Statistics} 

\author{Masahito Hayashi,$^\dagger$   {\em Fellow, IEEE},  $\,\,   \,$ Vincent Y.~F.\ Tan,$^\ddagger$ {\em Senior Member, IEEE}
\thanks{$^\dagger$Masahito~Hayashi is with the  Graduate School of Mathematics, Nagoya University, and the Center for Quantum Technologies (CQT),  National University of Singapore   (Email: masahito@math.nagoya-u.ac.jp).   } 
\thanks{$\ddagger$Vincent~Y.~F. Tan is with the Department of Electrical and Computer Engineering and the Department of Mathematics, National University of Singapore (Email:  vtan@nus.edu.sg).} \thanks{This paper was presented in part~\cite{HayashiTan17} at the 2017 International Symposium on Information Theory  (ISIT) in Aachen, Germany.}  \thanks{MH is partially supported  by a MEXT Grant-in-Aid for Scientific Research (B) No.~16KT0017.
MH is also partially supported by the Okawa Research Grant
and Kayamori Foundation of Informational Science Advancement. The Centre for Quantum Technologies is funded by the Singapore Ministry of Education and the National Research Foundation as part of the Research Centres of Excellence programme.}\thanks{VYFT is partially supported an NUS Young Investigator Award (R-263-000-B37-133) and a Singapore Ministry of Education Tier 2 grant ``Network Communication with Synchronization Errors: Fundamental Limits and Codes'' (R-263-000-B61-112). 
}  
   }


\maketitle

\begin{abstract} 
Given a  sufficient statistic for a parametric family of distributions, one can estimate the parameter without access to the data. However, the memory  or code size for storing the sufficient statistic may nonetheless still be prohibitive. Indeed,  for $n$ independent samples drawn from a $k$-nomial distribution with $d=k-1$ degrees of freedom, the length of the code  scales as $d\log n+O(1)$. In many applications, we may not have a useful notion of sufficient statistics (e.g., when the parametric family is not an exponential family) and we also may not need to reconstruct the generating distribution exactly. By adopting a Shannon-theoretic approach in which we allow a small error in estimating the generating distribution, we construct various   {\em approximate sufficient statistics} and show that the code length can be reduced to $\frac{d}{2}\log n+O(1)$. 
 We consider errors  measured according to the relative entropy and variational distance criteria. For the code constructions, we leverage Rissanen's minimum description length principle, which yields a non-vanishing error measured according to the relative entropy. For the converse parts, we use Clarke and Barron's formula for the relative entropy of a parametrized distribution and the corresponding mixture distribution.  
 However, this method only yields a weak converse for the variational distance. We develop new techniques to achieve vanishing errors and we also prove strong converses. The latter  means that even if the code is allowed to have a non-vanishing error, its length must still be at least~$\frac{d}{2}\log n$.
\end{abstract}  
 
\section{Introduction}


The notion of {\em sufficient statistics} is a fundamental  and ubiquitous concept in statistics and information theory~\cite{Lehmann98, Cov06}. Consider a random variable $X \in\calX$ whose distribution $P_{X|Z=z}$ depends on an unknown parameter $z\in\calZ$. Typically in detection and estimation problems, we are   interested in learning  the unknown parameter $z$. In this case, it is often unnecessary to use the full dataset $X$. Rather a function of the data $Y=f(X)\in\calY$ usually suffices. If there is no loss in the performance of learning $Z$ given $Y$ relative to the case when one is given  $X$, then $Y$ is called a sufficient statistic relative to the family $\{ P_{X|Z=z} \}_{z\in\calZ}$. We may then write
\begin{equation}
P_{X|Z=z}(x) = \sum_{y\in\calY} P_{X|Y}(x|y) P_{Y|Z=z}(y),\;\forall\, (x,z)\in\calX\times\calZ \label{eqn:mc_intro}
\end{equation}
or more simply  that 
$X \markov Y\markov Z$  forms a Markov chain in this order. Because $Y$ is a function of $X$, it is also true that $I(Z;X) = I(Z;Y)$.   This intuitively means that the sufficient statistic $Y$ provides as much information about the parameter $Z$ as the original data $X$ does.

For concreteness in our discussions, we often (but not always) regard the family $\{ P_{X|Z=z} \}_{z\in\calZ}$ as an exponential family~\cite{Wai08}, i.e., $P_{X|Z=z}\propto\exp\big( \sum_i z_i Y_i(x)\big)$. This class of distributions is parametrized by a set of natural parameters $z=\{z_i\}$ and a set of natural   statistics $Y(x)=\{Y_i(x)\}$, which is a function of the data. The natural   statistics or maximum likelihood estimator (MLE) are known to be sufficient statistics of the exponential family. In many applications, large datasets are prevalent. The one-shot model described above will then be  replaced by an $n$-shot one in which the dataset consists of $n$ independent and identically distributed (i.i.d.)  random variables $X^n = (X_1, X_2, \ldots, X_n)$ each distributed according to $P_{X|Z=z}$ where the exact $z$ is unknown. If the support of $X$ is finite, the distribution is a $k$-nomial distribution (a discrete distribution taking on at most $k$ values) and the     empirical distribution or type~\cite{Csi97} of $X^n$ is a sufficient statistic for learning $z$. However,  the number of types with denominator $n$ on an alphabet with $k$ values is $\binom{ n + k  -1}{k-1} = \Theta(n^{k-1})$~\cite{Csi97}. We are interested in this paper in the ``memory size'' to store the various types. We imagine that each type is allocated a single storage location in the memory stack and the {\em memory size} is the number of storage locations.  Thus, the memory size required to estimate parameter $z$ in a maximum likelihood manner (or distribution $P_{X|Z=z}$) is at least  $\Theta(n^{k-1})$ if the (index of the) type is stored.  The exponent $d=k-1$ here is the number of degrees of freedom in the distribution family, i.e., the dimensionality of the space $\calZ$ that $z$ belongs to. Can we do better than a memory size of $\Theta(n^{d})$? The answer to this question depends on the strictness of the recoverability condition of  $P_{X|Z=z}$. If $P_{X|Z=z}$ is to be recovered {\em exactly}, then the Markov condition in \eqref{eqn:mc_intro} is necessary and no reduction of the  memory size is possible.  However, if  $P_{X|Z=z}$ is to be recovered only {\em approximately}, we can indeed reduce the memory size. This is one motivation of the current work. 

In addition, going beyond exponential families, for general distribution families, we do not necessarily have a useful and universal notion of sufficient statistics.
Thus, we often focus on {\em local asymptotic sufficient statistics} by relaxing the condition for sufficient statistics.
For example, under  suitable regularity conditions~\cite{vanderVaart,LeCam,LeCam60}, the MLE forms a set of local asymptotic sufficient statistics.
However,   there is no prior work that discusses the required memory size if we allow the sufficient statistics to be approximate in some appropriate sense.
To address this issue, we introduce the notion of the minimum coding length of certain asymptotic sufficient statistics and show that it is $\frac{d}{2} \log n +O(1)$, where $d$ is the dimension of the parameter of the family of distribution. Hence, the minimum coding rate is the pre-log coefficient $\frac{d}{2}$, improving over the original $d$ when exact sufficient statistics are used. 
Here, we also notice that the locality condition can be dropped.
That is, our asymptotic sufficient statistics works globally. This is another motivation for the current paper.


\subsection{Related Work} \label{sec:related}
Our problem   is different from  lossy and lossless conventional source coding~\cite{Shannon48,Shannon59b} because we do not seek to reconstruct the data   $X^n$ but rather a distribution on $\calX^n$.  Hence, we need to generalize standard  data compression schemes.  Such a generalization has been discussed in the context of quantum data compression by Schumacher \cite{Schumacher95}.
Here,  the source that generates the state  cannot be directly observed. 
Schumacher's encoding process involves compressing the original dataset into a  memory stack with a smaller size. The decoding process involves recovering certain statistics of the data to within a prescribed error bound $\delta\ge 0$.  

 Reconstruction of distributions has also been studied in the context of the {\em information bottleneck (IB) method}~\cite{Tishby99,Chechik05,Harremores07}. In the IB method, given a joint distribution $P_{X,Y}$, one seeks to find the best tradeoff between accuracy and complexity when summarizing a random variable $X$ and an observed and correlated variable $Y$. More precisely, one finds a conditional   distribution $P_{\tilX|X}$ that minimizes $I(\tilX;X)-\beta I(\tilX;Y)$ where $\beta>0$ can be regarded as a Lagrange multiplier. The random variable $\tilX$ is then regarded as a summarized version of $X$.   Although such a formalism is a generalization of the notion of sufficient statistics from parametric statistics to arbitrary distributions, it differs from the present work because our work is concerned with finding minimum rates in an asymptotic and information-theoretic framework. 

Recently, Yang, Chiribella and Hayashi~\cite{Yang16} extended   Schumacher's \cite{Schumacher95} compression system to a special quantum model. In particular, the authors considered a notion of approximate sufficient statistics in the quantum setting~\cite{Petz86, Koashi} when the data is generated in an i.i.d.\ manner. They considered only the so-called blind setting~\cite[Ch.~10]{HayashiBook2017} and also only showed a weak converse.   We note that  there have been recent  developments of the notion of approximate sufficient statistics and approximate Markov chains in the quantum information literature~\cite{sutter16, fawzi15} but the problem studied here and the objectives are different from the existing works.

Another related line of works in the classical information theory literature are the seminal ones by Rissanen on universal variable-length source coding and model selection  \cite{Rissanen83,Rissanen84}.  Under the minimum description length  (MDL) framework, he introduced a two-step encoding   process to obtain a prefix-free source code for $n$ data samples generated from a mixture of  i.i.d.\ distributions. The   purpose of Rissanen's compression system is to obtain a  compression system for  data generated under a mixture distribution. He showed that when the dimensionality of the data is $d$, the optimal redundancy over the Shannon entropy is $\frac{d}{2}\log n +O(1)$. Merhav and Feder~\cite{merhav95} extended Rissanen's analysis to both the minimax and Bayesian (maximin) senses. Clarke and Barron~\cite{clarke90,clarke94} refined Rissanen's analysis and determined the constant (in   asymptotic expansions) under various regularity assumptions. While we make heavy use of some of Rissanen's coding ideas and Clarke and Barron's asymptotic expansions for the relative entropy between a parametrized distribution and a mixture, the problem setting we study is different. Indeed,  the main ideas in Rissanen's work \cite{Rissanen83,Rissanen84} are only helpful for us to establish the achievability parts of our theorems with {\em non-zero} asymptotic  error for the relative entropy criterion (see Lemma \ref{lem:mdl}). Similarly, the main results of Clarke and Barron's work~\cite{clarke90,clarke94} can only lead to a {\em weak converse} under the  variational distance criterion  (see Lemma~\ref{lem:clarke_barron}).  
Hence, we need to develop new coding techniques and converse ideas to satisfy the more stringent constraints on the code sequences.


\subsection{Main Contributions and Techniques} \label{sec:main_contr}

We provide a precise Shannon-theoretic problem formulation for   compression for the model parameter $z$ with an allowable asymptotic error $\delta\ge 0$ on the reconstructed distribution. This error is measured under the relative entropy and variational distance criteria.
  We use some of Rissanen's ideas for encoding in~\cite{Rissanen83,Rissanen84} to show that the memory size can be reduced   to approximately $\Theta(n^{\frac{d}{2}})$ resulting in a coding length of $\frac{d}{2}\log n +O(1)$.   Note that Rissanen~\cite{Rissanen83,Rissanen84}  did not explicitly provide the decoders for the problem he considered; we explicitly specify various decoders. Moreover, assuming that the parametric family of distributions is  an exponential family~\cite{Wai08}, we also improve on the evaluations that are inspired by Rissanen (see Lemma \ref{lem:exp}). In particular, for exponential families, we propose codes whose asymptotic errors measured according to the relative entropy criterion are equal to zero.  Furthermore, we consider two separate settings known as the {\em blind} and {\em visible} settings. In the former, the encoder can directly observe the dataset $X^n$; in the latter the encoder directly observes the parameter of interest $z$. The differences between these two settings are discussed in more detail in~\cite[Ch.~10]{HayashiBook2017}. The visible setting may appear to be less natural but such a generalized setting is useful for the proofs of the converse parts.   Yang, Chiribella and Hayashi~\cite{Yang16}  only considered the special case of the qubit model. They also only considered the blind setting.  
 We  consider {\em both} blind and visible settings and show, somewhat surprisingly, that the coding length is essentially unchanged.
  

Another significant contribution of our work is in the strengthening of the converse  in~\cite{Yang16}. 
In our strong converse proof for the relative entropy error criterion, we employ the Pythagorean theorem for relative entropy, a fundamental concept in information geometry~\cite{Ama00}.  Furthermore, we use Clarke and Barron's formula~\cite{clarke90, clarke94} to provide a weak converse under the variational distance error criterion. This clarifies the relation between our problem and  Clarke and Barron's formula~\cite{clarke90, clarke94}.  We significantly strengthen this method to obtain a strong converse (see Lemma \ref{lem:str}); in contrast~\cite{Yang16} only proves a weak converse.  That is, we show that  if the error is allowed to be non-vanishing (even if it is arbitrarily large for the relative entropy criterion  and arbitrarily close to $2$ for the variational distance criterion), we  would still require a memory size of at least $n^{d(\frac{1}{ 2}-\eta)}$ for any $\eta>0$ for all sufficiently large~$n$.

This paper is organized as follows. In Section \ref{sec:setup}, we formulate the problem precisely. We  state the main result Theorem \ref{thm:vis} in Section \ref{sec:main_res}. The results are discussed in the context of exact sufficient statistics and exponential families in Section~\ref{sec:suff_stat}. We prove the direct parts of Theorem \ref{thm:vis} in Section~\ref{sec:direct},   leveraging ideas from  Rissanen's seminal works~\cite{Rissanen83,Rissanen84} on universal data compression. We prove the converse parts of Theorem \ref{thm:vis} in Section \ref{sec:conv} by leveraging the Pythagorean theorem in information geometry~\cite{Ama00} and Clarke and Barron's formula~\cite{clarke90, clarke94}. We conclude our discussion in Section \ref{sec:discuss}. 
\section{Problem Formulation } \label{sec:setup}
Let $\calX$ be a   set and let $\calP(\calX)$ denote the set of distributions (e.g., probability mass functions) on $\calX$. We consider a family of distributions $\{P_{X|Z=z } \}_{z\in\calZ} \subset \calP(\calX)$ parametrized by a vector parameter $z\in\calZ\subset\bbR^d$. We assume that $n$ independent and identically distributed (i.i.d.) random variables $X^n = (X_1, \ldots, X_n)$ each taking values  in $\calX$ and drawn from $P_{X|Z=z}$. The underlying parameter $Z$, which is random,    follows a distribution $\mu(\rmd z)$, which is absolutely continuous with respect to the Lebesgue measure on $\bbR^d$.  We will often use the following notations: Given conditional distributions $P_{X|Y}$ and $P_{Y|Z}$, respectively let the joint and marginal probabilities conditioned on $Z=z$ be 
\begin{align}
(P_{X|Y} \times P_{Y|Z})(x,y|z)  &:=P_{X|Y}(x|y) P_{Y|Z}(y|z),\quad \mbox{and} \label{eqn:times_notation}\\
(P_{X|Y}\cdot P_{Y|Z})(x|z) &:= \sum_y (P_{X|Y} \times P_{Y|Z})(x,y|z).\label{eqn:dot_notation}
\end{align}
Standard asymptotic notation such as $o(\cdot)$, $O(\cdot)$, $\Omega(\cdot)$ and $\Theta$ will be used throughout; $f_n=o(g_n)$ iff $\varlimsup_{n\to\infty}|f_n/g_n|=0$, $f_n=O(g_n)$ iff $\varlimsup_{n\to\infty}|f_n/g_n|<\infty$, $f_n=\Omega(g_n)$ iff $ g_n = O(f_n)$ and $f_n=\Theta(g_n)$ iff $f_n=O(g_n)$ and $f_n=\Omega(g_n$). Standard information-theoretic notation such as entropy $H(\cdot)$ and mutual information $I(\cdot;\cdot)$  \cite{Cov06} will also be used. Finally, $\| \cdot\|$ and $\|\cdot\|_1$ denote the $\ell_2$ and $\ell_1$ norms of finite-dimensional vectors respectively. 
\subsection{Definitions of Codes} \label{sec:codes}
We consider two classes of codes~\cite[Ch.~10]{HayashiBook2017} for the problem of interest:

\begin{definition}[Blind code] \label{def:blind}
A size $M_n$  {\em blind code} of  $\calC_{\rmb,n}:=(f_{\rmb,n},\varphi_n)$    consists of  
\begin{itemize}
\item A  stochastic  encoder (transition kernel) $f_{\rmb,n} : \calX^n \to\calY_n := \{1,\ldots, M_n\}$;
\item A decoder $\varphi_n :\calY_n \to\calP(\calX^n)$.
\end{itemize}
\end{definition}

Observe that this definition of a code is similar to that for source coding except that the decoder outputs {\em distributions} on $\calX^n$ instead of length-$n$ strings in $\calX^n$.  We often consider a more relaxed condition for the encoder as follows. In the visible setting, the encoder does not only have access to the random vector $X^n$ but also to  the parameter $z\in\calZ$. 

\begin{definition}[Visible code]\label{def:vis}
A  size $M_n$ {\em visible code}  $\calC_{\rmv,n}:=(f_{\rmv,n},\varphi_n)$  consists of 
\begin{itemize}
\item A   stochastic  encoder (transition kernel) $f_{\rmv,n} : \calZ\to\calY_n := \{1,\ldots, M_n\}$;
\item A decoder $\varphi_n :\calY_n \to\calP(\calX^n)$.
\end{itemize}
\end{definition}

We note that any blind encoder $f_{\rmb,n}$ can be regarded as a special case of a visible encoder $f_{\rmv,n}$ because the visible encoder $f_{\rmv,n}$ can be written in terms of the blind encoder $f_{\rmb,n}$ and the distribution $P_{X|Z=z}^n$ as   follows
\begin{equation}
f_{\rmv,n}(z) :=\sum_{x^n \in\calX^n} f_{\rmb,n}(x^n) P_{X|Z=z}^n(x^n),\quad\forall\,z\in\calZ. \label{eqn:bv}
\end{equation}

\subsection{Error Criteria}\label{sec:error_crit}
The performance of any code is characterized by two quantities. First, we desire the coding length $\log M_n=\log |\calY_n|$ to be as short as possible.  Next we desire a small error. To define an error criterion precisely, we notice that the reconstructed distribution on $\calX^n$ (in the visible case) is $\varphi_n\cdot  f_{\rmv,n}(z)$ which is defined as  
\begin{align}
\!\!\!\big(\varphi_n\cdot  f_{\rmv,n}(z)\big)(x^n) \!= \!\sum_{y \in\calY_n}\Pr\left\{  f_{\rmv,n}(z)\! =\! y  \right\}\big(\varphi_n(y) \big)(x^n).\label{eqn:compose}
\end{align}
Hence the code has a smaller error when the distribution $\varphi_n\cdot  f_{\rmv,n}(z)$ is closer to the original distribution $P_{X|Z=z}^n$ for each $z\in\calZ$. To evaluate the difference between the two distributions, we consider an error or fidelity function $F$ whose inputs are distributions on the same probability space. The average error is defined as 
\begin{align}
\eps_\rmv( f_{\rmv,n},\varphi_n ) :=\int_{\calZ} F( \varphi_n\cdot  f_{\rmv,n}(z) , P_{X|Z=z}^n )\, \mu(\rmd z).  \label{eqn:vis_err}
\end{align}
For a blind code, in view of~\eqref{eqn:bv}, we similarly define 
\begin{align}
&\eps_\rmb( f_{\rmb,n},\varphi_n )   =\eps_\rmv( f_{\rmb,n}\cdot P_{X|Z=z}^n  ,\varphi_n ) \nn\\* 
&\qquad\qquad  :=\int_{\calZ} F( \varphi_n\cdot  f_{\rmv,n} \cdot P_{X|Z=z}^n , P_{X|Z=z}^n )\, \mu(\rmd z).  \label{eqn:bl_err}
\end{align}

In this paper, we consider two distance measures, namely the {\em relative entropy} $D(P\| Q) := \sum_x P(x)(\log P(x)-\log Q(x))$ and the {\em variational distance}\footnote{Unlike some papers, we define the variational distance without the coefficient of $\frac{1}{2}$ multiplying $\|P-Q\|_1$ so $0\le\|P-Q\|_1\le 2 $. }   (also known as the {\em total variation distance})  $\|P-Q\|_1:=\sum_x |P(x)-Q(x)|$. Generalizations of these ``distances'' to continuous-alphabet distributions are performed in the usual manner. We denote the errors in the blind and visible cases~\cite[Ch.~10]{HayashiBook2017} when we use the variational distance as $\eps_\rmb^{(1)}$ and $\eps_\rmv^{(1)}$ respectively. Similarly, we denote the errors in the blind and visible cases when we use the relative entropy as $\eps_\rmb^{(2)}$ and $\eps_\rmv^{(2)}$ respectively. The size of a code $\calC_{\rmb,n}$ is denoted as $|\calC_{\rmb,n}|=|\calY_n|$. 

\subsection{Definitions of Minimum Compression Rates and   Properties}
\begin{definition}[Minimum  Compression Rate]
Let $\delta\ge 0$. We define the {\em minimum  compression rate for blind codes} for a given parametric family $\{ P_{X|Z=z} \}_{z\in\calZ}$ as 
\begin{align}
R_{\rmb}^{ (i) }(\delta) :=\inf_{ \{\calC_{\rmb,n} \}_{n\in\bbN }} \left\{ \varlimsup_{n\to\infty} \frac{\log  |\calC_{\rmb,n} | }{\log n} \,\bigg|\, \varlimsup_{n\to\infty}\eps_\rmb^{(i)}(\calC_{\rmb,n})\le\delta   \right\} \label{eqn:Rb}
\end{align}
where $i=1,2$ denotes whether the error function is the variational distance or relative entropy respectively. In a similar manner, we define the {\em minimum compression rate  for visible codes} for a given parametric family $\{ P_{X|Z=z} \}_{z\in\calZ}$ as 
\begin{align}
R_{\rmv}^{ (i) }(\delta) :=\inf_{ \{\calC_{\rmv,n} \}_{n\in\bbN }} \left\{ \varlimsup_{n\to\infty} \frac{\log  |\calC_{\rmv,n} | }{\log n} \,\bigg|\, \varlimsup_{n\to\infty}\eps_\rmv^{(i)}(\calC_{\rmv,n})\le\delta   \right\}.\label{eqn:Rv}
\end{align}
\end{definition}
To understand this definition, we note that if $R_{\rmb}^{ (i) }(\delta) = c>0$, then  for every $\epsilon>0$,  there exists a sequence of codes  $\{\calC_{\rmb,n}\}_{n\in\bbN}$ with asymptotic error no larger than $\delta$ and memory or coding length upper bounded as 
$|\calC_{\rmb,n} |\le  n^{c+\epsilon}$ 
 for $n$ large enough. Moreover, there is no sequence of codes with with asymptotic error no larger than $\delta$ and with $|\calC_{\rmb,n} |\le n^{c-\epsilon}$.  
 
 The definition of the minimum compression rate differs significantly from traditional source coding in Shannon theory~\cite{Cov06} where the normalization of the coding length $\log|\calC_{\rmb,n}|$ is $n$  and not $\log n$. Here, we find that the  normalization that yields meaningful results is $\log n$ as the memory size scales polynomially  (and not exponentially) with the blocklength, i.e., $|\calY_n|\approx n^c$ for some $c>0$.  
 
From the above definitions, it is clear that for any $0\le\delta\le\delta'$ and $i = 1,2$, we have 
\begin{align}
R_{a}^{ (i) }(\delta') &\le R_{a}^{ (i) }(\delta), \quad \,\, a = \rmb,\rmv,\label{eqn:smaller_delta}\\
R_{a}^{ (1) }(0) &\le R_{a}^{ (2) }(0), \quad\, a = \rmb,\rmv, \label{eqn:use_pinsker}\\
R_{\rmv}^{ (i) }(\delta) &\le R_{\rmb}^{ (i) }(\delta). \label{eqn:visible_blind}
\end{align}
Note that \eqref{eqn:use_pinsker} follows from Pinsker's inequality (i.e., $D(P\| Q)\ge\frac{\log\rme}{2}\|P-Q\|_1^2$)  because a vanishing relative entropy implies the same for the  variational distance.

\section{Assumptions and Main Results} \label{sec:main_res}
Let $J_z$ be the Fisher information matrix of the   parametric family $\{ P_{X|Z=z } \}_{z\in\calZ}$. This matrix has elements
\begin{equation}
[J_z]_{i,j} \!=\! \bbE_z\left[  \left(\frac{\partial \log P_{X|Z=z}(X)}{\partial z_i }\right)\! \left(\frac{\partial\log P_{X|Z=z}(X)}{\partial z_j }\right)\right], \label{eqn:fisher}
\end{equation}
where $\bbE_z$ means that we take expectation with respect to $X\sim P_{X|Z=z}$.
   Before we state the main results of this paper, we consider the following assumptions:
\begin{enumerate}
\item[(i)] {\em (Boundedness of Parameter Space)} The set $\calZ \subset\bbR^d$ is bounded and has positive Lebesgue measure in $\bbR^d$;
\item[(ii)]  {\em (Euclidean Approximation of Relative Entropy)} As $z'\to z$, the relation 
\begin{align}
&D(P_{X|Z=z} \| P_{X|Z=z'} )   \nn\\*
&= \frac{1}{2}\sum_{i,j} [J_z]_{i,j} (z_i-z_i')(z_j-z_j') + o( \| z-z'\|^2) \label{eqn:euc_re}
\end{align}
holds~\cite{Bor08,Abbe10,Ama00}.   We also assume compact convergence (i.e., uniform convergence on compact sets) for \eqref{eqn:euc_re}.
\item[(iii)]  {\em (Asymptotic Efficiency)} There exists a sequence of estimators $\hatz_n=\hatz_n(X^n)$ for the parameter $z$ such that 
\begin{equation}
\bbE_{z}\left[ D(P_{X| Z=\hatz_n} \| P_{X|Z=z}) \right] = \frac{d}{2n} + o\Big( \frac{1}{n} \Big). \label{eqn:div_expand}
\end{equation}
In other words, the estimator $\hatz_n$ asymptotically achieves the Cram\'er-Rao lower bound~\cite{Lehmann98,Ibragimov81} (i.e.,  $\bbE_z[ (\hatz_n-z)(\hatz_n-z)^T ] \to J_z^{-1}$), so the expectation of~\eqref{eqn:euc_re} with $z=z_n$ and $z'=z$   yields \eqref{eqn:div_expand}.
\item[(iv)]  {\em (Local Asymptotic Normality)}  Fix a point $z\in\calZ$ and let $\hatz_{\mathrm{ML}}(X^n)$ be the MLE of $z$ given $X^n$.  Define the function  $h_z(X^n) = \sqrt{n}J_z^{1/2} (\hatz_{\mathrm{ML}}(X^n) - z)$ and let $\phi^{(d)}(x): = (2\pi)^{-d/2} \exp(-\|x\|^2/2)$ be the $d$-dimensional standard Gaussian probability  density function. The {\em local asymptotic normality} condition  \cite{LeCam,vanderVaart,LeCam60} reads
\begin{equation}
\Big\|\phi^{(d)} - \Big( P_{X|Z=z+\frac{z'}{\sqrt{n}}}^n \cdot h_{z+ \frac{z'}{\sqrt{n}}}^{-1} \Big) \Big\|_1\to 0\label{eqn:local_asymp_norm}
\end{equation}
for any vector $z'\in\bbR^d$.
\item[(v)] {\em (Local Asymptotic Sufficiency)}  Let $Z$ be the random variable corresponding to the parameter $z$ and let $Y'$ be the corresponding MLE $\hatz_{\mathrm{ML}}(X^n)$. The {\em local  asymptotic sufficiency} condition  \cite{LeCam,vanderVaart,LeCam60} reads
\begin{equation}
\Big\|  \Big( P_{X^n |Y',Z=z}\cdot P_{Y'|Z=z+ \frac{z'}{\sqrt{n}}}  \Big)- P_{X |Z= z+ \frac{z'}{\sqrt{n}}}^n  \Big\|_1\to 0 \label{eqn:local_asymp_suff}
\end{equation}
for any vector $z'\in\bbR^d$.
\end{enumerate}
Conditions  (i)--(v) are  satisfied as long as  the parametrized  family  satisfies some weak smoothness condition~\cite{LeCam,vanderVaart,LeCam60}.
\begin{theorem} \label{thm:vis}
Assuming (i), (ii), (iv) and (v), the minimum compression rate  for visible codes under the variational distance error criterion  
\begin{align}
R_{\rmv}^{ (1) }(\delta_1)  &=\frac{d}{2} ,\quad  \forall\, \delta_1\in [0,2).\label{eqn:vis1}  
\end{align}
Assuming (i), (ii),  the minimum compression rate  for visible codes under the relative entropy error criterion  
\begin{align}
R_{\rmv}^{ (2) }(\delta_2)  &= \frac{d}{2}, \quad  \forall\, \delta_2\in [0,\infty)\label{eqn:vis2}  .
\end{align}
Assuming (i), (ii), (iv), and (v), the minimum compression rate   for blind codes under the variational distance error criterion  
\begin{align}
R_{\rmb}^{ (1) }(\delta_1') &=\frac{d}{2},\quad \forall\,\delta_1'\in [0,2).\label{eqn:blind1} 
\end{align}
Assuming (i), (ii), and (iii)    the minimum compression rate   for blind codes under the relative entropy error criterion  
\begin{align}
R_{\rmb}^{ (2) }(\delta_2')  & = \frac{d}{2}, \quad  \forall\, \delta_2'\in \Big[\frac{d}{2},\infty\Big). \label{eqn:blind2}
\end{align}
Furthermore, if (i) holds and $\{ P_{X|Z=z} \}_{z\in\calZ}$ is an exponential family~\cite{Wai08}, \eqref{eqn:blind2} can be strengthened to 
\begin{align}
R_{\rmb}^{ (2) }(\delta_2')   = \frac{d}{2} ,\quad\forall\, \delta_2'\in [0,\infty). \label{eqn:blind2_exp}
\end{align}
\end{theorem}
 The direct and converse parts of this theorem are proved in Sections \ref{sec:direct} and \ref{sec:conv} respectively. Remarks on and implications of the theorem are detailed in the following section.

\section{Connection to Sufficient Statistics and Exponential Families} \label{sec:suff_stat}
In this section, we discuss the implications of Theorem \ref{thm:vis} in greater detail by relating them to the notion of (exact) sufficient statistics~\cite[Sec.~2.9]{Cov06}. We first review the fundamentals of sufficient statistics, then motivate the notion of approximate sufficient statistics, provide some background on exponential families, and finally show that if one stores the exact sufficient statistics in the memory $\calY_n$, the memory size would be larger than that prescribed by Theorem \ref{thm:vis}.
\subsection{Review of Sufficient Statistics} \label{sec:ss}
 Suppose, for the moment, that the blind encoder $f_{\rmb,n}$ is a deterministic function. When $Y  = f_{\rmb,n}(X^n)$ is a {\em sufficient statistic} relative to the family $\{ P_{X|Z=z}\}_{z\in\calZ}$ \cite[Sec.~2.9]{Cov06},  the conditional distribution $P_{X|Z=z, Y=y}^n(x^n)$ does not depend on $z$, i.e., $Z\markov Y\markov X$
forms a Markov chain in this order. In this case, we can choose the decoder $\varphi_n:\calY_n\to\calP(\calX^n)$  as follows
\begin{equation}
\varphi_n(y) := P_{X|Z=z, Y=y}^n   . \label{eqn:dec_ss}
\end{equation}
 Now, noting that $P_{X|Z=z}^n ( \{x^n: f_{\rmb,n}(x^n) =y\}) = f_{\rmb,n} \cdot P_{X|Z=z}^n(y)$ for every $y$ in the memory $\calY_n$, we have 
 \begin{align}
&\varphi_n\cdot f_{\rmb,n}    =\sum_{y}  f_{\rmb,n}\cdot P_{X|Z=z}^n(y)\varphi_n(y)\\
 &\quad=\sum_{y}  P_{X|Z=z}^n ( \{x^n: f_{\rmb,n}(x^n) =y\})  P_{X|Z=z, Y=y}^n\\
 &\quad=P_{X|Z=z}^n. \label{eqn:dec_ss1}
 \end{align}
Observe that, in this case,  regardless of which error metric we choose, we will attain zero error between the reconstructed distribution  $\varphi_n\cdot f_{\rmb,n}$ and the original one $P_{X|Z=z}^n$. However, as we will see, if $P_{X|Z=z}$ is an exponential family the memory size exceeds that prescribed by the various statements in Theorem \ref{thm:vis}. Thus it is natural to relax the stringent condition in~\eqref{eqn:dec_ss} to some approximate versions of it.
  
 \subsection{Exact vs.\ Approximate Sufficient Statistics}
To consider an approximate version of \eqref{eqn:dec_ss}, we will make an  assumption on the error function $F$ 
\begin{enumerate}
\item[(*)] Consider distributions $P_i$ and $P_i'$ such that they respectively have disjoint supports from $P_j, j \ne i$ and $P_j' , j \ne i$. Then we assume that 
\begin{equation}
F\bigg( \sum_i p_i P_i ,\sum_i p_i P_i' \bigg)=\sum_i p_i F(  P_i , P_i'  ) \label{eqn:cond4}
\end{equation}
where $\{p_i\}$ forms a probability mass function. This is clearly satisfied for the variational distance error function. 
\end{enumerate}
Then for $\delta\ge 0$, we have
\begin{align}
&\eps_\rmb(f_{\rmb,n},\varphi_n) \nn\\*
&= \int_\calZ F( \varphi_n \cdot f_{\rmb,n} \cdot P_{X|Z=z}^n, P_{X|Z=z}^n)\, \mu(\rmd z)\\
&= \int_\calZ F\bigg(\sum_y (f_{\rmb,n} \cdot P_{X|Z=z}^n )(y) \varphi_n(y) , \nn\\*
&\qquad \sum_y ( f_{\rmb,n} \cdot P_{X|Z=z}^n) (y) P_{X|Z=z,Y=y}^n  \bigg)    \, \mu(\rmd z)\\
&= \int_\calZ \sum_y( f_{\rmb,n} \cdot P_{X|Z=z}^n )(y) \nn\\*
&\qquad\times  F\big(  \varphi_n(y) ,  P_{X|Z=z,Y=y}^n   \big)    \, \mu(\rmd z)   \label{eqn:use_cond4}\\
&\le\delta,\label{eqn:err}
\end{align}
where \eqref{eqn:use_cond4} uses Condition (*) of $F$ in \eqref{eqn:cond4}, and \eqref{eqn:err} uses the fact that the error is bounded by $\delta$ according to  \eqref{eqn:Rb} and \eqref{eqn:Rv}. If $\delta=0$, the equality in \eqref{eqn:dec_ss} holds, and we revert to the usual notion of exact sufficient statistics discussed in Section~\ref{sec:ss}. Hence, the codes that we consider allowing for errors can be regarded as a generalization of sufficient statistics.

\subsection{Background on Exponential Families} \label{sec:exp_fam}
To put our results in Theorem \ref{thm:vis} into context, we regard $\{ P_{X|Z=z} \}_{z\in\calZ}$ as an exponential family~\cite{Wai08} with parameter space $\calZ\subset\bbR^d$. Recall that a  parametric family of distributions $\{ P_{X|Z=z} \}_{z\in\calZ}$ is called an {\em exponential family}  if it takes the form
\begin{equation}
P_{X|Z=z}(x) = P_X(x) \exp\left[  \sum_{i=1}^m  z_i Y_i(x)-A(z) \right], \label{eqn:exp_fam}
\end{equation}
where $A(z)$, the cumulant generating function    of the random vector $(Y_1(X),\ldots, Y_m(X))$,  is defined as
\begin{equation}
A(z) := \log\sum_x P_X(x) \exp\left[ \sum_{i=1}^m z_i Y_i(x)\right].\label{eqn:log_par}
\end{equation}
The functions $Y_i(x)$ are known as the {\em sufficient statistics} of the exponential  family. Another fact that we exploit in the sequel is that for any exponential family, there is an alternative parametrization known as the {\em moment parametrization} \cite{Wai08}. There is a one-to-one correspondence between the natural  parameter $z$ and the expectation parameter
\begin{equation}
\eta_i(z) := \frac{\partial A(z)}{\partial z_i}=\bbE_z[Y_i]=\sum_x P_{X|Z=z}(x) Y_i(x).
\end{equation}
Hence, in the following to estimate the natural parameter $z$, we can first estimate the moments $\eta(z) = (\eta_1(z),\ldots, \eta_m(z))\in \calH:=\{ \eta(z): z\in\calZ\}$ and then use the  one-to-one correspondence to obtain $z$. 

\subsection{An Example: $k$-nomial Distributions} \label{sec:knomial}
Now as a concrete example, we consider a $k$-nomial distribution, i.e., the family of discrete distributions that take on $k\in\bbN$ values. The set of $k$-nomial distributions forms an exponential family with sufficient statistics
\begin{align}
Y_i(x) = \left\{ \begin{array}{cc}
1 & x=i \\
-1 & x=i+1 \\
0 & \mathrm{else}
\end{array} \right. ,\quad i = 1, \ldots, k-1 . \label{eqn:Yi}
\end{align}
Note that there are other parametrizations. It is known that the vector  $Y(x):=(Y_1(x),\ldots,Y_{k-1}(x))$ allows us to recover information about the unknown parameter $z$ \cite{Ama00,Ama01}, i.e., $Y(x)$ is a sufficient statistic for the $k$-nomial distribution.  

Given $n$ i.i.d.\ data samples from $P_{X|Z=z}^n$, the exponential family can be written as 
\begin{equation}
 P_{X|Z=z}^n(x^n)\! =\! P_X^n(x^n)\exp\left[ \sum_{i=1}^m Y_i^{(n)}(x^n) z_i \!-\! n A(z) \right]\! ,\! 
\end{equation}
where $x^n=(x_1,\ldots, x_n)\in\calX^n$ and $Y_i^{(n)}(x^n) := \sum_{j=1}^n Y_i(x_j)$. Thus the vector of sufficient statistics is $Y^{(n)}(x^n):= (Y_1^{(n)}(x^n),\ldots, Y_m^{(n)}(x^n))$.  In the $k$-nomial case, the dimension of the exponential family $m =k-1$. It is easy to see that the total number of possibilities of $Y^{(n)}(x^n)$, i.e., the size of the set $\{ Y^{(n)}(x^n) : x^n\in\calX^n\}$ is $\binom{ n+k-1}{k-1}$. This is also the total number of $n$-types \cite{Csi97} on an alphabet of size $k$. In this case, the required memory size is 
\begin{equation}
 \log |\calY_n| = \log \binom{ n+k-1}{k-1}= (k-1)\log n + O(1).
 \end{equation} 
Note  that $k-1=d$,  the dimension of the parameter space $\calZ$. Thus,  the pre-log coefficient is $d$, which is twice as large as what the results of Theorem \ref{thm:vis} prescribe if we use approximate sufficient statistics  in the sense of \eqref{eqn:err} instead of exact sufficient statistics discussed in Section \ref{sec:ss}. This motivates us to study the fundamental limits of approximate sufficient statistics in the large $n$ limit to   reduce the memory size $|\calY_n|$ from   $n^{d+o(1)}$ to   $n^{\frac{d}{2}+o(1)}$.
\section{Proofs of Direct Parts of Theorem \ref{thm:vis}} \label{sec:direct}
In this section, the direct parts (upper bounds) of Theorem~\ref{thm:vis} will be proved. For logical reasons, the statements in Theorem~\ref{thm:vis} will not be proved sequentially. Rather we will present the simplest proofs before proceeding to the proofs for more general statements.    First, in Section~\ref{sec:mdl}, we will prove semi-direct part   for the relative entropy criterion in~\eqref{eqn:blind2}. This immediately leads the proof of the direct part for~\eqref{eqn:vis2}.  Next, in Section~\ref{sec:exp_prf}, we strengthen the direct part for exponential families under the relative entropy criterion in~\eqref{eqn:blind2_exp}. Finally, in Section~\ref{sec:gen},  we prove the  direct part  in the blind setting under the variational distance criterion as in \eqref{eqn:blind1}. 

\subsection{Semi-Direct Part Based On Rissanen's Minimum Description Length (MDL) Encoder } \label{sec:mdl}

Here we prove the direct parts for \eqref{eqn:vis2} and \eqref{eqn:blind2} where the error criterion used   is the relative entropy. We present a complete achievability proof in the  visible setting, i.e., \eqref{eqn:vis2}.
Notice that \eqref{eqn:vis2} implies \eqref{eqn:vis1}.
Under the same error criterion, we show  a semi-achievability in the  blind setting  (i.e., \eqref{eqn:blind2}) in which
the error does not vanish even in the limit of large~$n$.

\begin{lemma} \label{lem:mdl}
Assuming (i), (ii), we have 
\begin{align}
R_\rmv^{(2)}(0) \le\frac{d}{2}. \label{eqn:direct_v2}
\end{align}
In addition, assuming (i), (ii), and (iii),
\begin{align}
R_\rmb^{(2)}\Big(\frac{d}{2}\Big) \le\frac{d}{2}.\label{eqn:direct_b2}
\end{align}
\end{lemma}
Note that \eqref{eqn:direct_v2} proves that  the direct part of~\eqref{eqn:vis2} holds because it implies that $R_\rmv^{(2)}(\delta_2) \le \frac{d}{2}$ for all $\delta_2\in [0,\infty)$. Similarly, \eqref{eqn:direct_b2} proves that~\eqref{eqn:blind2} holds because it implies that  $R_\rmb^{(2)}(\delta_2') \le \frac{d}{2}$ for all $\delta_2'\in [\frac{d}{2},\infty)$.  These statements follow immediately from the bound in \eqref{eqn:smaller_delta} concerning the monotonicity of $\delta\mapsto R_\rmv^{(2)}(\delta )$ and  $\delta\mapsto R_\rmb^{(2)}(\delta )$. 

We also   note that \eqref{eqn:direct_b2}, which follows from Rissanen's ideas~\cite{Rissanen83,Rissanen84}, is rather weak because the asymptotic error is bounded above by $\frac{d}{2}$ instead  of  $0$.  We improve on this severe  limitation in the subsequent subsections.

\begin{proof}[Proof of Lemma \ref{lem:mdl}]
We first prove \eqref{eqn:direct_b2}. Then we describe how to modify the argument slightly to show~\eqref{eqn:direct_v2}. Fix a lattice span $t>0$ and consider the subset $\calZ_{n,t} := \frac{t}{\sqrt{n}}\bbZ^d\cap \calZ\subset\calZ$. Given the MLE $\hatz_n: =\hatz_{\mathrm{ML}}(X^n)$, we consider the closest point  
\begin{equation}
z_{n,t}(\hatz_n):=\argmin_{z' \in \calZ_{n,t} } \sum_{i,j } [J_z]_{i,j} (\hatz_{n,i}-z_i')(\hatz_{n,j}-z_j'). \label{eqn:approx_grid}
\end{equation}
That is, for this blind encoder, the memory $\calY_n$ is  taken to be $\calZ_{n,t} $ and the encoder is $f_{\rmb,n}(X^n):= z_{n,t}(\hatz_{\mathrm{ML}}(X^n))$, i.e., we first compute the  MLE then we approximate it with a point in a finite subset $\calZ_{n,t}$ using the formula in  \eqref{eqn:approx_grid}.  The decoder is the map from the parameter $z_{n,t}$ to the distribution $P_{X|Z=z_{n,t}}^n$. The coding length is $\log |\calZ_{n,t}|= \frac{d}{2}\log n + O(1)$, where the dependence on $t$ is in the $O(1)$ term. Note that Rissanen~\cite{Rissanen83, Rissanen84} essentially proposed the same encoder but he was considering a different problem of universal source coding. Also Rissanen did not explicitly specify the decoder. We also mention that Merhav and Feder~\cite{merhav95} extended Rissanen's analysis to both the minimax and Bayesian (maximin) formulations. 

Now,  for any $r>0$ and any norm $\|\cdot\|$, we have the inequality
$\|a-b\|^2 \le (1+r)\|a\|^2  + (1+\frac{1}{r})\|b\|^2$, a consequence of the basic fact that $\big\|\sqrt{r}a-\sqrt{\frac{1}{r} }b\, \big\|^2\ge 0$. 
Applying this inequality to the norm $\frac{1}{2}\| \cdot \|_{J_z}$ ($J_z$ is positive definite) with $a\equiv\hatz_n - z$ and $b\equiv  \hatz_n-z_{n,t}$, we obtain
\begin{align}
&\frac{1}{2}\sum_{i,j} [ J_z ]_{i,j} (z_{n,t,i}- z_i)(z_{n,t,j}- z_j) \nn\\*
&\quad\le \frac{1+r}{2} \sum_{i,j} [ J_z ]_{i,j} (\hatz_{n,i}- z_i)(\hatz_{n,j}- z_j) \nn\\*
&\qquad + \frac{1+\frac{1}{r}}{2} \sum_{i,j} [ J_z ]_{i,j} (z_{n,t,i}- \hatz_{n,i})(z_{n,t,j}- \hatz_{n,j}). \label{eqn:cs_ineq}
\end{align}
We now estimate the error as follows:
\begin{align}
& \eps_{\rmb}^{(2)} ( f_{\rmb,n}, \varphi_n)  \nn\\*
 & := \int_{\calZ } D \left( \bbE_z[ P_{X |Z= z_{n,t}}^n ] \,\big\|\, P_{X|Z=z}^n\right)\, \mu(\rmd z) \label{eqn:expectat} \\
& \le \int_\calZ \bbE_z \left[ D \left( P_{X |Z= z_{n,t}}^n \,\big\|\,  P_{X|Z=z}^n  \right) \right]\, \mu(\rmd z) \label{eqn:use_jens}\\
& = \int_\calZ n\bbE_z \bigg[ \frac{1}{2}\sum_{i,j} [J_z]_{i,j} (z_{n,t,i}- z_i)(z_{n,t,j}- z_j)\nn\\*
&\quad+ o( \| z_{n,t}-z\|^2)  \bigg] \, \mu(\rmd z)\label{eqn:use_div_expansion} \\
&\le \int_\calZ\bbE_z \bigg[ \frac{n (1+r)}{2} \sum_{i,j} [ J_z ]_{i,j} (\hatz_{n,i}- z_i)(\hatz_{n,j}- z_j) \nn\\*
&\quad + \frac{n(1+\frac{1}{r})}{2} \sum_{i,j} [ J_z ]_{i,j} (z_{n,t,i}- \hatz_{n,i})(z_{n,t,j}- \hatz_{n,j}) \nn\\*
&\quad+ o( \| z_{n,t}-z\|^2)  \bigg]\, \mu(\rmd z)\label{eqn:use_cs_ineq}\\
&=   \frac{1  +  r }{2}   \int_{\calZ} d\, \mu(\rmd z) +o(1) \nn\\*
&\quad+ \frac{ n(1+ \frac{1}{r}) }{2} \int_\calZ\bbE_z \bigg[   \sum_{i,j} [ J_z ]_{i,j} (z_{n,t,i}- \hatz_{n,i})(z_{n,t,j}- \hatz_{n,j}) \nn\\*
&\quad   +   o\big(  {\| z_{n,t}-z\|^2}  \big)  \, \mu(\rmd z) \bigg] \label{eqn:use_def_lattice} \\
&\le   \frac{1  +  r }{2}   \int_{\calZ} d\, \mu(\rmd z) +o(1) + \frac{   1+ \frac{1}{r}   }{2}    \sum_{i\ge j} |[ J_z ]_{i,j}|t^2\nn\\*
&\quad+  \int_\calZ\bbE_z \left[  o\big( n {\| z_{n,t}-z\|^2}  \big) \right]  \, \mu(\rmd z)  \label{eqn:use_def_lattice2} \\ 
  &= \frac{(1+r )d}{2}  + \frac{ 1+ \frac{1}{r} }{2} \sum_{i\ge j } | [J_z]_{i,j} | t^2  + o(1) .\label{eqn:use_dim}
\end{align}
We now justify  some of the steps above. In  \eqref{eqn:expectat}, the expectation is over the random $z_{n,t}(\hatz_{\mathrm{ML}}(X^n))$   where $X^n\sim P_{X|Z=z}^n$; in~\eqref{eqn:use_jens} we used Jensen's inequality and the convexity of the relative entropy; in \eqref{eqn:use_div_expansion} we used the Euclidean approximation of the relative entropy in~\eqref{eqn:euc_re} (Assumption (ii)); in \eqref{eqn:use_cs_ineq} we used the inequality in \eqref{eqn:cs_ineq};   in~\eqref{eqn:use_def_lattice}  we used~\eqref{eqn:div_expand}  (Assumption (iii)); and   in \eqref{eqn:use_def_lattice2} and \eqref{eqn:use_dim} we used the definition of the lattice $\calZ_{n,t}$ resulting in the bound $|z_{n,t,i}- \hatz_{n,i}|\le  \frac{t }{\sqrt{n}}$ for all~$i$ and~$n$.

Now since $n\in\bbN$ is arbitrary,
\begin{equation}
\varlimsup_{n\to\infty}\eps_{\rmb}^{(2)} ( f_{\rmb,n}, \varphi_n)\le\frac{1+r}{2} d +  \frac{1+\frac{1}{r}}{2} \sum_{i\ge j } | [J_z]_{i,j} | t^2   \label{eqn:asymp_error_b2}
\end{equation}
Since $t>0$ is arbitrary, we may take $t\to 0$ so the second term vanishes. Next, since $r>0$ is arbitrary, we may take $r\to 0 $ so the first term converges to the asymptotic error   bound of $\frac{d}{2}$. This proves the upper bound in \eqref{eqn:direct_b2}. 

We now consider visible case, which is simpler. In this case, we can replace the MLE $\hatz_n$ by $z$, since the encoder has direct access to the parameter $z$. Hence, the first terms in~\eqref{eqn:use_cs_ineq}--\eqref{eqn:asymp_error_b2} are equal  $0$ and we obtain~\eqref{eqn:direct_v2} as desired.\end{proof}
\subsection{Direct Part For Exponential Families } \label{sec:exp_prf}

In the blind setting, the MDL encoder discussed in Section~\ref{sec:mdl} has a non-vanishing error even in the asymptotic limit.  
To overcome this problem, we devise a novel method attaining zero error in the asymptotic limit.  
Since the method is more complicated in the general setting,  and requires more assumptions (see Section \ref{sec:gen}), 
we first assume that the distribution family forms an exponential family,
and prove the direct part under the relative entropy criterion.

\begin{lemma} \label{lem:exp}
When $\{P_{X|Z=z}\}_{z\in\calZ}$ is an exponential family and Assumption (i) holds, we have 
\begin{align}
R_\rmb^{(2)}(0)  \le\frac{d}{2}.
\end{align}
\end{lemma}
Note that this indeed implies the direct part of~\eqref{eqn:blind2_exp} since   $R_\rmb^{(2)}(\delta_2')  \le \frac{d}{2}$ for all $\delta_2'\in [0,\infty)$. This significantly improves over the case where we do not assume that $\{P_{X|Z=z}\}_{z\in\calZ}$ is an exponential family in \eqref{eqn:direct_b2} of  Lemma \ref{lem:mdl} since we could only prove that $R_\rmb^{(2)}(\delta_2')  \le \frac{d}{2}$ for all $\delta_2'\in [\frac{d}{2},\infty)$, which is much weaker. In other words, the blind code presented below for exponential families can realize the same error   performance  (which is  asymptotically zero) as the visible code  presented at the end of the proof of Lemma~\ref{lem:mdl}. We also observe that Assumption (i) holds for the $k$-nomial example discussed in Section \ref{sec:knomial} as the moment parameters $\bbE_z[T_i]$ belong  to $[-1,1]$, which is bounded.

\begin{proof}[Proof of Lemma \ref{lem:exp}]
First, to describe the encoder, we extract the sufficient statistics  from the data, i.e., we calculate
\begin{equation}
\hat{\eta}_i := \frac{ Y_{i}^{(n)}(x^n)}{n}  = \frac{1}{n}\sum_{j=1}^n Y_i(x_j)  ,\quad\forall\, i = 1,\ldots, d.
\end{equation}
Next we fix a lattice span $t>0$ and consider the subset of quantized moment parameters $\calH_{n,t} := \frac{t}{\sqrt{n}} \bbZ^d \cap\calH\subset\calH$ where recall that $\calH = \{\eta (z) : z\in\calZ\}$ is the set of feasible moment parameters (cf.~Section \ref{sec:exp_fam}). Given the observed value of  $\hat{\eta} = (\hat{\eta}_1,\ldots, \hat{\eta}_n)$, we choose the closest point in the lattice to it, i.e.,  we choose
\begin{equation}
\beta_t(\eta):=\argmin_{\eta'\in\calH_{n,t}} \|\eta'-\eta\|. \label{eqn:encoding_map}
\end{equation}
The encoder  $f_{\rmb,n}$ is the map $x^n \mapsto\beta_t(\hat{\eta})$. 
For this encoder, the memory size is $|\calH_{n,t}|$. Thus the coding length is $\log |\calH_{n,t}|=\frac{d}{2}\log n + O(1)$, where the dependence in $t$ is in the $O(1)$ term. 

Now, we describe the decoder.  Let $Y^{(n)} = (Y^{(n)}_1 ,\ldots, Y^{(n)}_d)$ and $\eta=(\eta_1,\ldots,\eta_d)$. The decoder $\varphi_n$ is the map   
\begin{equation}
\bar{\eta}  \mapsto \frac{1}{|\beta_t^{-1}(\bar{\eta} )|}\sum_{ \eta \in \beta_t^{-1}(\bar{\eta} )} P_{X^n |   Y^{(n)}  =n\eta  } .
\end{equation}
In other words, given the estimate $\beta_t(\eta) \in \calH_{n,t}$, we consider a uniform mixture of all the distributions $ P_{X^n |  Y^{(n)}  =n\eta  }$ where $\eta$ runs over all points in the lattice that map to $\beta_t(\eta)$ under the encoding map in \eqref{eqn:encoding_map}.

In the following calculation of the error, we first consider the scalar case in which $d=1$ for simplicity. At the end of the proof, we show how to extend the ideas to the case where $d>1$. A few additional notational conventions are needed. For each element  $\bar{\eta}$, denote the uniform distribution on the subset $\beta_t^{-1}(\bar{\eta} )$ as $U_{\beta_t^{-1}(\bar{\eta} )}$. Next, denote the transition kernel (channel) that maps   the mean parameter $\bar{\eta}$ to the uniform distribution on the set $\beta_t^{-1}(\bar{\eta} )$ (i.e.,  $U_{\beta_t^{-1}(\bar{\eta} )}$) as $U_{\beta_t^{-1}( Y) |Y}$. Denote the distribution of the random variable $Y_{1}^{(n)}(X^n)$ when $X^n\sim P_{X|Z=z}^n$ as $P_{ Y_{1}^{(n)}  |Z=z}$.  Let the variance of $Y_{1}^{(n)}$ under distribution $P_{X|Z=z}$ be $V_z$. The normalizing linear transformation $y\mapsto \sqrt{n}(y-\bbE_z[Y_{1}^{(n)}])/\sqrt{V_z}$ is denoted as $g_z$.

\begin{figure}[t]
\centering
\includegraphics[width=.995\columnwidth]{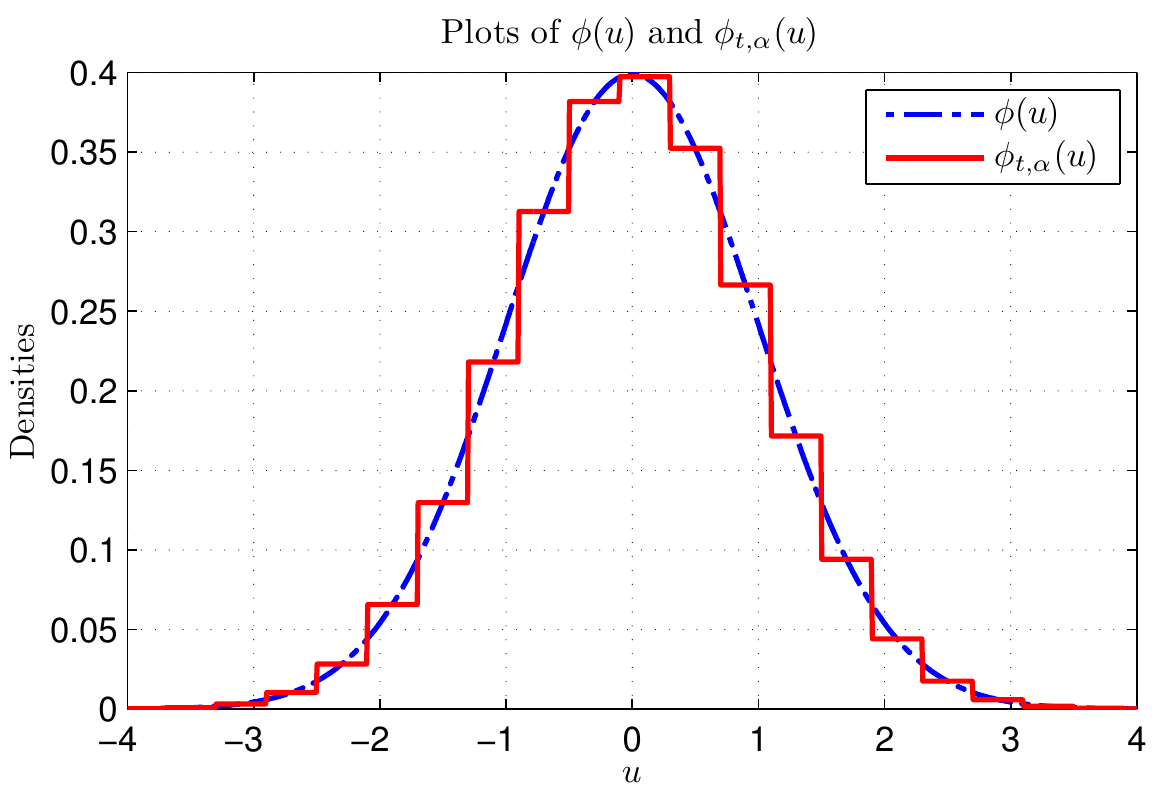}
\caption{Illustration of the density $\phi_{t,\alpha}$ in \eqref{eqn:phi_t_alpha} with $t = 0.4$ and $\alpha=0.3$ (solid red line). The standard normal density $\phi$ is also shown (broken blue line).  The two curves become increasingly close as $t\to 0$.}
\label{fig:quantize}
\end{figure}
Since $Y_{1}^{(n)}$ is a sufficient statistic relative to the exponential family $\{P_{X|Z=z}\}_{z\in\calZ}$, we know that  for any error criterion (in particular the relative entropy criterion),
\begin{align}
&D( \varphi_n\cdot f_{\rmb,n} \cdot P_{X|Z=z}^n  \,\| \, P_{X|Z=z}^n)\nn\\*
 &=D( U_{\beta_t^{-1}( Y ) |Y} \cdot P_{Y_{1}^{(n)}|Z=z} \,\| \,P_{Y_{1}^{(n)}|Z=z} )\\
&= D \big( (U_{\beta_t^{-1}( Y ) |Y} \cdot P_{Y_{1}^{(n)}|Z=z})\cdot g_z^{-1}\,\| \, P_{Y_{1}^{(n)}|Z=z }\cdot g_z^{-1} \big),
\end{align}
where the last equality follows from the fact that the function $g_z$ is one-to-one. By the central limit theorem, $P_{Y_{1}^{(n)}|Z=z }\cdot g_z^{-1}$ converges to the standard normal distribution $\phi(u) \propto \exp(-u^2/2)$. On the other hand, by the definition of $U_{\beta_t^{-1}( Y ) |Y} $ (which results from the construction of the encoder in \eqref{eqn:encoding_map}), the distribution $(U_{\beta_t^{-1}( Y ) |Y} \cdot P_{Y_{1}^{(n)}|Z=z})\cdot g_z^{-1} $ converges to  a quantization of the standard normal distribution with span $t$, namely,
\begin{equation}
\phi_{t,\alpha} (u ) \propto \phi\big(\alpha + (j+ 1/2 ) t\big) ,\;\forall\,u\in ( \alpha+jt,\alpha+ (j+1)t] \label{eqn:phi_t_alpha} ,
\end{equation}
where $\alpha\in [0,t]$ and the constant of proportionality in \eqref{eqn:phi_t_alpha}  is chosen  so that $\int_\bbR\phi_{t,\alpha}(u)\,\rmd u=1$. See Fig.~\ref{fig:quantize} for an illustration of the probability  density function in \eqref{eqn:phi_t_alpha}. 
Thus, we have the upper bound
\begin{align}
&\varlimsup_{n\to\infty} D\big( (U_{\beta_t^{-1}( Y ) |Y} \cdot P_{Y_{1}^{(n)}|Z=z})\cdot g_z^{-1} \,\| \, P_{Y_{1}^{(n)}|Z=z }\cdot g_z^{-1} \big)\nn\\*
&\qquad\le\sup_{\alpha \in [0,t]} D( \phi_{t,\alpha} \,\| \, \phi).   \label{eqn:compact_conv}
\end{align}
Since the   convergence in \eqref{eqn:compact_conv} is  uniform on compact sets (compact convergence), we have 
\begin{align}
&\varlimsup_{n\to\infty}\eps_{\rmb}^{(2)}(f_{\rmb,n}, \varphi_n)\nn\\*
&=\varlimsup_{n\to\infty}\int_{\calZ}D( \varphi_n\cdot f_{\rmb,n}\cdot P_{X|Z=z}^n\,\| \,P_{X|Z=z}^n)\, \mu(\rmd z) \label{eqn:compact_convergence}\\
&\le \int_{\calZ}\varlimsup_{n\to\infty} D( \varphi_n\cdot f_{\rmb,n}\cdot P_{X|Z=z}^n\,\| \, P_{X|Z=z}^n)\, \mu(\rmd z)\\
&\le\sup_{\alpha \in [0,t]} D( \phi_{t,\alpha} \,\| \, \phi).   \label{eqn:sup_alpha}
\end{align}
Now since the above holds for all $t>0$, we can let $t$ tend  to $0$ (so the size of the quantization regions decreases to $0$). Consequently, the right-hand-side of \eqref{eqn:sup_alpha} also tends to $0$ and hence,
$
\lim_{n\to\infty}\eps_{\rmb}^{(2)}(f_{\rmb,n}, \varphi_n)=0$.

For  general dimension $d>1$, we can show the desired statement as follows. Let $\phi^{(d)}$ be the $d$-dimensional standard normal distribution. Given $\alpha\in [0,t]^d$,  let  $\phi^{(d)}_{t,\alpha}$  be the corresponding quantization of the $d$-dimensional standard normal distribution with cutting point  $\alpha$ and span $t$ (cf.~\eqref{eqn:phi_t_alpha} for the one-dimensional distribution). Then in the same way,
\begin{equation}
\varlimsup_{n\to\infty}\eps_{\rmb}^{(2)}(f_{\rmb,n}, \varphi_n)\le\sup_{\alpha\in [0,t]^d} D(  \phi_{t,\alpha}^{(d)}\,\| \,\phi^{(d)}).
\end{equation}
Similarly, we can take $t$ to tend to zero and the error criterion vanishes as $n\to\infty$. The logarithm of the memory size (coding length) is thus $\frac{d}{2}\log n +O(1)$. This proves Lemma \ref{lem:exp}.
\end{proof}

\subsection{Direct Part For The General Case}\label{sec:gen}
In this section we treat the general case (not necessarily exponential family).  We prove the following lemma which establishes the direct part  (upper bound) in the blind setting under the variational distance criterion as in \eqref{eqn:blind1}. 
\begin{lemma} \label{lem:gen}
Assuming (i), (ii), (iv), and (v), we have
\begin{align}
R_\rmb^{(1)}(0)  \le\frac{d}{2}. \label{eqn:gen}
\end{align}
\end{lemma}
Note that this implies the upper bound to \eqref{eqn:blind1} because \eqref{eqn:gen} implies that $R_\rmb^{(1)}(\delta_1')  \le\frac{d}{2}$ for all $\delta_1' \in [0,2)$.  
\begin{proof}[Proof of Lemma \ref{lem:gen}]
Fix a lattice span $t>0$ and choose the memory $\calY_n$ to be the quantized parameter space (lattice) $\calZ_{n,t} := \frac{t}{\sqrt{n}}\bbZ^d \cap\calZ\subset\calZ$. Given the observed MLE $\hatz_{\mathrm{ML}}(X^n)=z$, we choose  the encoder output to be the closest point in this lattice, i.e.,
\begin{align}
\beta_t(z):= \argmin_{z' \in \calZ_{n,t}} \| z'-z\|. \label{eqn:beta_t_z}
\end{align}
The encoder  $f_{\rmb,n}$ is the map from $x^n\mapsto\beta_t(\hatz_{\mathrm{ML}}(x^n))$.  Thus the code has memory $\calY_n=\calZ_{n,t}$ in which the coding length is $\log |\calY_n|=\log|\calZ_{n,t}|=\frac{d}{2}\log n +O(1)$. 

Now, to describe the decoder and the subsequent analysis, we use some simplified notation. Let $Y$ and $Y'$ denote the random variables $\beta_t(\hatz_{\mathrm{ML}}(X^n))$ and  $\hatz_{\mathrm{ML}}(X^n)$ respectively.  These can be thought of as the quantized MLE and the MLE respectively. As usual $Z\in\calZ$ is the original parameter. The decoder $\varphi_n$ is then the following map from elements in the memory to distributions in $\calP(\calX^n)$:
\begin{equation}
\barz\mapsto\frac{1}{|\beta_t^{-1}(\barz)|} \sum_{z\in \beta_t^{-1}(\barz)}  P_{X^n | Y'=z, Y=\barz}.
\end{equation}
Essentially, the decoder takes the  quantized MLE $\beta_t(\hatz_{\mathrm{ML}}(X^n))$  and outputs a uniform mixture over all ``compatible'' conditional distributions (i.e., all conditional distributions $P_{X^n | Y'=z, Y=\barz}$ whose parameter $z$ lies in the quantization cell $\beta_t^{-1}(\barz)$).

We  now estimate the error. We first consider the case $d=1$ for simplicity. For each element $\barz$, denote the uniform distribution on the subset $\beta_t^{-1}(\barz)$ as $U_{Y'| Y=\barz}$. We denote the transition kernel corresponding to the map $\barz\mapsto U_{Y'| Y=\barz}$ as $P_{X^n | Y', Y}$. Then the decoder $\varphi_n$ can alternatively be written as the map $z\mapsto P_{X^n | Y',Y=z}  \cdot U_{Y'| Y=z}$.  Now the error  measured according to the  variational distance can be written as 
\begin{align}
&\eps_{\rmb}^{(1)}( f_{\rmb,n}, \varphi_n)  \nn\\*
&=  \int_\calZ\left\| \varphi_n\cdot f_{\rmb,n} \cdot P_{X^n|Z=z} - P_{X^n|Z=z} \right\|_{1}\,\mu(\rmd z)\\
&=\int_\calZ \big\| (P_{X^n | Y',Y }\cdot  U_{Y'| Y}) \times P_{Y = \beta_t (\hatz_{\mathrm{ML}}(X^n)) | Z=z  } \nn\\*
&\qquad- P_{X^n|Z=z}\big\|_{1}\,\mu(\rmd z)\label{eqn:def_enc_dec}
\end{align}
where \eqref{eqn:def_enc_dec} follows from the definitions of the encoder $P_{Y = \beta_t (\hatz_{\mathrm{ML}}(X^n)) | Z=z }$ (given the parameter is $Z=z$) and decoder $P_{X^n | Y',Y=z}  \cdot U_{Y'| Y=z}$. For clarity, we write $P_{Y = \beta_t (\hatz_{\mathrm{ML}}(X^n)) | Z=z }$  for the distribution of the quantized MLE $\beta_t (\hatz_{\mathrm{ML}}(X^n))$ given that the samples $X^n$ are independently generated from the distribution parametrized by $z\in\calZ$, i.e., $P_{X|Z=z}^n$.  Let the integrand in \eqref{eqn:def_enc_dec} for fixed $z$ be denoted as  $\eps_{\rmb}^{(1)}( f_{\rmb,n}, \varphi_n;z) $. By the triangle inequality, 
\begin{equation}
\eps_{\rmb}^{(1)}( f_{\rmb,n}, \varphi_n;z)\le A_n+B_n, \label{eqn:ab}
\end{equation}
where  the sequences
$A_n$ and $B_n$ are defined as
\begin{align}
A_n & := \big\| (P_{X^n | Y',Y }\cdot  U_{Y'| Y}) \times P_{Y = \beta_t (\hatz_{\mathrm{ML}}(X^n)) | Z=z  }   \nn\\*
&\qquad -  P_{X^n | Y',Y }\cdot P_{Y=\beta_t(\hatz_{\mathrm{ML}}(X^n)), Y' = \hatz_{\mathrm{ML}}(X^n) | Z=z} \big\|_1, \label{eqn:An} \end{align}
and 
\begin{align}
 B_n &:=\big\|  P_{X^n | Y',Y }\cdot P_{Y=\beta_t(\hatz_{\mathrm{ML}}(X^n)), Y' = \hatz_{\mathrm{ML}}(X^n) | Z=z} \nn\\*
&\qquad - P_{X^n|Z=z}\big\|_{1}.  \label{eqn:Bn_def}
\end{align}
Note that $P_{Y=\beta_t(\hatz_{\mathrm{ML}}(X^n)), Y' = \hatz_{\mathrm{ML}}(X^n) | Z=z}$ is the joint distribution of the quantized MLE $\beta_t(\hatz_{\mathrm{ML}}(X^n))$ and the true MLE $\hatz_{\mathrm{ML}}(X^n)$ given that the samples $X^n$ are independently generated from the distribution parametrized by $z\in\calZ$.  
Now,  by the data processing inequality for the variational distance  (i.e., $\| P_{X|Y}\cdot P_{Y|Z=z}- P_{X|Y}\cdot Q_{Y|Z=z}\|_1\le\|   P_{Y|Z=z}-  Q_{Y|Z=z}\|_1$),  the term $A_n$ in \eqref{eqn:An} can be bounded as 
\begin{align}
A_n &\le \big\| U_{Y'| Y} \times P_{Y = \beta_t (\hatz_{\mathrm{ML}}(X^n)) | Z=z  }  \nn\\*
&\quad-    P_{Y=\beta_t(\hatz_{\mathrm{ML}}(X^n)), Y' = \hatz_{\mathrm{ML}}(X^n) | Z=z} \big\|_1  \label{eqn:bd_An}
\end{align} 
 We now analyze the right-hand-sides of  \eqref{eqn:bd_An} and \eqref{eqn:Bn_def} in turn.

By a similar reasoning as in the proof of Lemma~\ref{lem:exp}, $U_{Y'| Y} \times P_{Y = \beta_t (\hatz_{\mathrm{ML}}(X^n)) | Z=z  }$  converges to the quantization of the standard normal distribution $\phi_{t,\alpha}$   (defined in \eqref{eqn:phi_t_alpha}) by the local asymptotic normality   assumption as stated in \eqref{eqn:local_asymp_norm} (Assumption~(iv)). Furthermore, since $Y$ is a deterministic function of $Y'$,  we have the  relation $  P_{Y=\beta_t(\hatz_{\mathrm{ML}}(X^n)), Y' = \hatz_{\mathrm{ML}}(X^n) | Z=z}(y,y') = P_{Y' = \hatz_{\mathrm{ML}}(X^n) | Z=z}(y)\bone\{ y=\beta_t(y')\} $. Hence,   the distribution $  P_{Y=\beta_t(\hatz_{\mathrm{ML}}(X^n)), Y' = \hatz_{\mathrm{ML}}(X^n) | Z=z}$ converges to the    standard normal distribution $\phi$ again by the   local asymptotic normality   assumption as stated in \eqref{eqn:local_asymp_norm} (Assumption (iv)).  Applying the triangle inequality   to the right-hand-side of \eqref{eqn:bd_An}, we have 
\begin{align}
&\varlimsup_{n\to\infty}A_n  \nn\\* 
&\le \varlimsup_{n\to\infty}\Big\{\big \| U_{Y'| Y} \times P_{Y = \beta_t (\hatz_{\mathrm{ML}}(X^n)) | Z=z  }   -\phi_{t,\alpha} \big\|_1 \nn\\*
&\qquad+ \big\| \phi_{t,\alpha}-\phi\big\|_1   \nn\\*
&\qquad +\big\| \phi- P_{Y=\beta_t(\hatz_{\mathrm{ML}}(X^n)), Y' = \hatz_{\mathrm{ML}}(X^n) | Z=z} \big\|_1\Big\}\\
 &\le \sup_{\alpha \in [0,t]}\big \| \phi_{t,\alpha}-\phi\big\|_1.
\end{align}

Now we analyze the right-hand-side of \eqref{eqn:Bn_def}. First, note that  $Y'-z = \hatz_{\mathrm{ML}}(X^n) - z $ behaves as $\Theta(\frac{1}{\sqrt{n}})$ with   probability  tending to one by the central limit theorem (local asymptotic normality). Since  the quantization level is also of the order $\Theta(\frac{1}{\sqrt{n}})$, the   difference $Y-z=\beta_t(\hatz_{\mathrm{ML}}(X^n)) - z$ also behaves as $\Theta(\frac{1}{\sqrt{n}})$ with probability  tending to one. 
Hence,  by   regarding $Y$ as the random variable $\bone\{Z=\tilz\}$ for some $\tilz\in\bbR^d$ that differs from $z\in\calZ\subset\bbR^d$   by  $\Theta(\frac{1}{\sqrt{n}})$,  we may write 
\begin{align}
B_n  \le\big\|  P_{X^n | Y'  , Z=\tilz}\cdot P_{ Y' = \hatz_{\mathrm{ML}}(X^n) | Z=z}- P_{X^n|Z=z}\big\|_{1} .\label{eqn:bn_ub}
\end{align}
At this point, we may apply  the local asymptotic sufficiency assumption as stated in \eqref{eqn:local_asymp_suff} (Assumption (v)) to \eqref{eqn:bn_ub}, yielding
\begin{equation}
\lim_{n\to\infty}B_n=0.
\end{equation}

By \eqref{eqn:ab}, and similar compact convergence arguments as those leading from \eqref{eqn:compact_convergence} to \eqref{eqn:sup_alpha},  we find that
\begin{equation}
\varlimsup_{n\to\infty} \eps_{\rmb}^{(1)}( f_{\rmb,n}, \varphi_n)\le\sup_{\alpha \in [0,t]}\big \| \phi_{t,\alpha}-\phi\big\|_1.
\end{equation}
Since this statement holds for all $t>0$, taking the limit $t\to 0$,   we see that the asymptotic error $\lim_{n\to\infty}\eps_{\rmb}^{(1)}( f_{\rmb,n}, \varphi_n) =0$. So in the general case for $d=1$, we can achieve a  memory length (log of memory size or coding length) of $\frac{1}{2}\log n + O(1)$. 

The case in which $d>1$ can be analyzed in a completely analogous manner and we can conclude that a  memory length of $\frac{d}{2}\log n + O(1)$ can be achieved. 
\end{proof}
\section{Proofs of Converse Parts of Theorem \ref{thm:vis}} \label{sec:conv}

In this section, we prove the converse parts (lower bounds)  to Theorem \ref{thm:vis}. 
We will only focus on the visible cases in \eqref{eqn:vis1} and \eqref{eqn:vis2} because
according to \eqref{eqn:visible_blind}, a converse for the visible case implies the same for the blind case. Essentially, by \eqref{eqn:bl_err}, a visible code cannot be outperformed by a blind code.

Since our problem is closely related to Clarke and Barron's formula for the relative entropy between a parametrized distribution and a mixture
distribution~\cite{clarke90, clarke94}, we clarify the relation between our problem and this formula.
To clarify this relation, in Section \ref{sec:clarke_barron}, we prove a weak converse, namely, 
 the impossibility of further compression  from a rate of $\frac{d}{2}$ when the variational distance error criterion is asymptotically zero. This 
can be shown by a simple combination of Clarke and Barron's formula and the uniform continuity of mutual information~\cite{Zhang07} (also called Fannes inequality~\cite{Fannes73} in quantum information).
Since the variational distance goes to zero when the relative entropy goes to zero,
the weak converse under  the variational distance criterion implies 
the weak converse under the relative entropy criterion.
Hence, the arguments in Section~\ref{sec:clarke_barron}   demonstrate the weak converse under both error criteria.
These arguments  clarify the relation between Clarke and Barron's formula and our problem.
However, to the best of our knowledge, the strong converse parts cannot be shown via Clarke and Barron's formula, i.e., they require novel methods.
Furthermore, there is no similar relation between the strong converse parts under the variational distance and the relative entropy.
This is because there is no relation between code rates when the  relative entropy is arbitrarily large and  when 
the variational distance is arbitrarily close to $2$, i.e., its maximum value.
So, we need to prove two types of strong converse parts for each of the two error criteria. 
In Section~\ref{sec:conv_re}, we prove a strong converse for the relative entropy error criterion
 using the Pythagorean theorem for the relative entropy, thus demonstrating~\eqref{eqn:vis2}. 
In Section \ref{sec:str_conv}, we prove a strong converse for the variational distance error  criterion by a different, and novel, method, thus demonstrating~\eqref{eqn:vis1}.


\subsection{Weak Converse Under Both Criteria Based On Clarke And Barron's Formula} \label{sec:clarke_barron}

In this section, we prove the following weak converse.
\begin{lemma} \label{lem:clarke_barron}
The following lower bound holds
\begin{equation}
R_{\rmv}^{(1)}(0)\ge\frac{d}{2}. \label{eqn:clarke_barron_res}
\end{equation}
\end{lemma}
This is, in fact, only a weak converse since asymptotically the error measured according to the variational distance must tend to zero. It is insufficient to show \eqref{eqn:vis1} but we present the proof to demonstrate the connection between Clarke and Barron's result in  \eqref{eqn:clarke_barron} to follow and the problem we study. Here, we are only concerned with the variational distance criterion because a weak converse for this criterion implies the same for the relative entropy criterion. 

\begin{proof}[Proof of Lemma \ref{lem:clarke_barron}]
 We first assume that $\calX$ is a finite set. At the end, we show how to relax this condition. We recall that  Clarke and Barron~\cite{clarke90,clarke94} showed for a parametric family $\{ P_{X|Z=z} \}_{z\in\calZ}$ that 
\begin{align}
&\int_\calZ  D\bigg(   P_{X|Z=z}^n \,\Big\|\, \int_\calZ  P_{X|Z=z'}^n \, \nu(\rmd  z')\bigg)\, \mu(\rmd z) \nn\\*
&= \frac{d}{2}\log\frac{n}{2\pi \rme} + D(\mu\|\nu) - D(\mu\| \mu_\rmJ) +\log C_\rmJ+o(1), \label{eqn:clarke_barron}
\end{align}
where $J_z$ is the Fisher information matrix defined in \eqref{eqn:fisher}, $\mu_\rmJ(\rmd z) 
:= \frac{1}{C_\rmJ} \mathrm{det}\sqrt{J_z} \, \rmd z$ is the so-called {\em Jeffrey's prior} \cite{clarke94} and 
$C_\rmJ := \int_\calZ \mathrm{det}\sqrt{J_z}\, \rmd z$ is the normalization factor. 
When $\nu=\mu$, the left-hand-side of~\eqref{eqn:clarke_barron} is precisely the mutual information $I(X^n;Z)$ where the pair of random variables $(X^n,Z)$ is distributed according to $P_{X^n,Z}(x^n,z):=P_{X|Z=z}^n(x^n)\mu(z)$. See~\cite{Barron98} for an overview of approximations similar to \eqref{eqn:clarke_barron} in the context of universal source coding and   model   selection. 

For the purpose of proving the weak converse, we assume that we are given a sequence of codes $\{\calC_{\rmv,n} := ( f_{\rmv,n}, \varphi_n) \}_{n\in\bbN}$ satisfying  the condition that the error measured according to the variational distance vanishes, i.e.,
\begin{equation}
\delta_n:=\eps_\rmv^{(1)}(\calC_{\rmv,n})\to 0 ,\quad\mbox{as  } n \to\infty. \label{eqn:weak_conv}
\end{equation}
Now let the code distribution be $P_{\calC_{\rmv,n}}(x^n,z) := ( \varphi_n\cdot f_{\rmv,n}(z))(x^n) \mu(z)$ where  $ \varphi_n\cdot f_{\rmv,n}(z)$ is defined in \eqref{eqn:compose} and $( \varphi_n\cdot f_{\rmv,n}(z))(x^n) $ is the evaluation of $\varphi_n\cdot f_{\rmv,n}(z)$ at $x^n$.  Then,  the definition of the error in the visible case in~\eqref{eqn:vis_err} and~\eqref{eqn:weak_conv} implies that  the variational distance between the code distribution and the generating distribution satisfies
\begin{equation}
\left\| P_{\calC_{\rmv,n}} - P_{X^n,Z} \right\|_1 \le\delta_n, \label{eqn:delta_n1}
\end{equation}
for some sequence $\delta_n=o(1)$.
Since $|\calX|$ is finite, we can use the method of types to find a set of sufficient statistics for the data (cf.~Section \ref{sec:knomial}). Indeed, we can form a set of sufficient statistics relative to the family $\{ P_{X|Z=z} \}_{z\in\calZ}$. Let us call the sufficient statistics $G_n :\calX^n\to\bbR^{|\calX|}$. The output cardinality of $G_n $ is $|G_n(\calX^n) | = | \{ G_n(x^n): x^n\in\calX^n \}| \le (n+1)^{|\calX|-1}$ because we can take the type of $x^n$ to be the sufficient statistic, i.e., $G_n(x^n) = \mathrm{type}(x^n)$.  By the data processing inequality for  the variational distance, we have 
\begin{equation}
\left\|P_{\calC_{\rmv,n}}\cdot G_n^{-1} - P_{X^n,Z}\cdot G_n^{-1}  \right\|_1 \le\delta_n.  \label{eqn:dpi_var}
\end{equation}
In the following we use a subscript to denote the distribution of the random variables in the arguments of the mutual information functional, so for example $I_{P_{AB}}(A;B) = \sum_a P_A(a) D( P_{B|A}(\cdot |a ) \| P_B)$. Now, we notice that
\begin{align}
I_{ P_{X^n,Z} \cdot G_n^{-1}}(G_n(X^n) ; Z) & = I_{ P_{X^n,Z} }( X^n  ; Z) ,\quad\mbox{and} \label{eqn:use_ss}\\
I_{ P_{\calC_{\rmv,n}} \cdot G_n^{-1}}(G_n(X^n) ; Z) & \le I_{ P_{\calC_{\rmv,n}}  }( X^n  ; Z)  \label{eqn:use_dpi} 
\end{align}
where \eqref{eqn:use_ss} follows from the definition of sufficient statistics and~\eqref{eqn:use_dpi} follows from the data processing inequality for mutual information. In addition, by the uniform continuity of mutual information~\cite{Zhang07, Fannes73} and \eqref{eqn:dpi_var}, we have
\begin{align}
&\big|I_{ P_{X^n,Z} \cdot G_n^{-1}}(G_n(X^n) ; Z)-  I_{ P_{\calC_{\rmv,n}} \cdot G_n^{-1}}(G_n(X^n) ; Z) \big| \nn\\*
&\quad \le\delta_n \log( (n+1)^{|\calX|-1}  ) + \xi(\delta_n), \label{eqn:cont_MI}
\end{align}
where $\xi(x) :=-x\log x$. We define the upper bound between the two mutual information quantities in~\eqref{eqn:cont_MI} as $\delta_n':= \delta_n \log( (n+1)^{|\calX|-1}  ) + \xi(\delta_n)$ and note that $\delta_n'=o(\log n)$. 

Define the joint  distribution of the encoder and the parameter as $P_{f_{\rmv, n}}(y,z) := P_{f_{\rmv, n}(z)}(y) \mu(z)$ (recall the code is  visible so $f_{\rmv,n} $ has access to $z$). Consider  the mutual information between the parameter $Z$ and the memory index~$Y$, 
\begin{align}
&\log |\calY_n|   \nn\\*
&\ge I_{P_{f_{\rmv, n}}}(Y;Z)  \label{eqn:mi_bd}\\
& = \int_\calZ D\left( f_{\rmv,n}(z) \,\Big\|\,\int_\calZ f_{\rmv,n}(z') \, \mu(\rmd z')  \right)\, \mu(\rmd z )\label{eqn:def_mi}\\
& \ge \int_\calZ \! D\left( \!\varphi_n\cdot  f_{\rmv,n}(z)  \Big\| \int_\calZ \varphi_n\cdot  f_{\rmv,n}(z') \mu(\rmd z')  \right) \mu(\rmd z)\label{eqn:dpi_kl} \\
&= I_{ P_{\calC_{\rmv,n}}} (X^n;Z)\label{eqn:def_mi2} \\
&\ge I_{ P_{\calC_{\rmv,n}} \cdot G_n^{-1} } (G_n(X^n);Z) \label{eqn:dpiMI}\\
&\ge I_{ P_{X^n,Z} \cdot G_n^{-1} } (G_n(X^n);Z) - \delta_n'\label{eqn:change_dist}\\
& =  I_{ P_{X^n,Z}   } (X^n ;Z) - \delta_n' \label{eqn:use_suff_stats} \\
&= \int_\calZ D\left( P_{X|Z=z}^n \,\Big\|\, \int_\calZ P_{X|Z=z'}^n\, \mu(\rmd z' )  \right) \,\mu(\rmd z)- \delta_n'\label{eqn:def_mi3} \\
&=  \frac{d}{2}\log\frac{n}{2\pi \rme}   - D(\mu\| \mu_\rmJ) +\log C_\rmJ+o(1)- \delta_n'\label{eqn:use_cb}\\
&=\frac{d}{2}\log n + o(\log n).
\end{align}
In the above chain,  \eqref{eqn:def_mi}, \eqref{eqn:def_mi2}, and \eqref{eqn:def_mi3} follow from the definition of mutual information, \eqref{eqn:dpi_kl} follows from the data processing inequality for the relative entropy, \eqref{eqn:dpiMI} follows from the data processing inequality for mutual information,~\eqref{eqn:change_dist} follows from the uniform continuity of mutual information as stated in~\eqref{eqn:cont_MI}, \eqref{eqn:use_suff_stats} follows from the notion of sufficient statistics as seen in~\eqref{eqn:use_ss}, and  \eqref{eqn:use_cb} follows from Clarke and Barron's formula~\cite{clarke90} with $\nu=\mu$ in  \eqref{eqn:clarke_barron}.  We conclude that if a sequence of codes is such that the variational distance vanishes as in~\eqref{eqn:weak_conv}, the memory size $|\calY_n|\ge n^{ \frac{d}{2} + o(1)}$. 

Now, when $\calX$ is not a finite set, we can choose a finite disjoint partition $\{ \calS_w \}_{w\in\calW}$ of $\calX$ satisfying the following conditions: (i) $|\calW|$ is finite  and (ii) $\cup_{w\in\calW} \calS_w= \calX$. Now, we define the parametric family $P_{W|Z=z}(w) := P_{X|Z=z}(\calS_w)$. Clearly, we can   go through the above proof with the finite-support random variable $W$ in place of $X$. Now, when the code reconstructs the original family $\{ P_{X|Z=z}^n\}_{z\in\calZ}$, clearly it also reconstructs  the quantized family $\{ P_{X|W=w}^n\}_{w\in\calW}$. In essence, reconstructing the latter is ``easier'' than the former. Since \eqref{eqn:clarke_barron_res} holds for the family $\{ P_{X|W=w}^n\}_{w\in\calW}$ it must also hold for  $\{ P_{X|Z=z}^n\}_{z\in\calZ}$. This completes the proof of \eqref{eqn:clarke_barron_res}.
\end{proof}

\subsection{Strong Converse Under The Relative Entropy Criterion} \label{sec:conv_re}

In this section, we prove the following  strong converse result using the Pythagorean theorem for the relative entropy.
\begin{lemma} \label{lem:pyt}
The following lower bound holds
\begin{equation}
R_{\rmv}^{(2)}(\delta_2)\ge\frac{d}{2},\quad\forall\, \delta_2\in [0,\infty). \label{eqn:pyt_res}
\end{equation}
\end{lemma}
This proves the lower bound to \eqref{eqn:vis2}.  The proof hinges on the Pythagorean formula for the relative entropy and a  geometric argument also contained in Rissanen's work~\cite{Rissanen84}.
\begin{proof}[Proof of Lemma \ref{lem:pyt}]
Given probability measures $\{P_i\}_{i\in\calI}$ and $Q$, and a probability mass function $\{p_i\}_{i\in\calI}$, the Pythagorean formula for relative entropy~\cite{Ama00} states that 
\begin{equation}
\!\sum_{i\in\calI} p_i D(P_i \| Q) \!=\! D  \bigg( \sum_{i\in\calI} p_i P_i \Big\| Q\bigg) \!+\! \sum_{i\in\calI} p_i D\bigg( P_i  \Big\| \sum_{j\in\calI} p_j P_j\bigg) . \label{eqn:pyt}
\end{equation}
In the following, we show that if the memory size $|\calY_n|$ is too small, say $n^{ \frac{d}{2}(1-\epsilon)}$ for some fixed $\epsilon>0$, then the error $\eps_{\rmv}^{(2)} (\calC_{\rmv,n})$ tends to infinity as $n$ grows. 

Consider any code $\calC_{\rmv,n}= (f_{\rmv,n},\varphi_n)$ with memory size $|\calY_n| =n^{ \frac{d}{2}(1-\epsilon)}$.  Let $P_{f_{\rmv,n}(z)}(y)  = \Pr\{ f_{\rmv,n}(z)=y\}$ be the probability that the index in the memory $Y\in\calY_n$ takes on the value $y$ when the parameter is $z\in\calZ$ under the random encoder  mapping $f_{\rmv,n}$. Then an application of the  Pythagorean theorem  in \eqref{eqn:pyt} yields
\begin{align}
&D\bigg( \sum_{y\in\calY_n} P_{f_{\rmv,n}(z)}(y) \varphi_n(y) \bigg\| P_{X|Z=z}^n \bigg) \nn\\*
& =\sum_{y\in\calY_n} P_{f_{\rmv,n}(z)}(y) D( \varphi_n(y) \| P_{X|Z=z}^n)  \nn\\*
&\quad- \sum_{y\in\calY_n} \! \! P_{f_{\rmv,n}(z)}(y)  D\bigg(\!  \varphi_n(y) \bigg\| \sum_{y'\in\calY_n} \!\!  P_{f_{\rmv,n}(z)}(y')\varphi_n(y')  \bigg). \label{eqn:apply_pyt}
\end{align}
Hence, by integrating  \eqref{eqn:apply_pyt} over all $z\in\calZ$, we obtain
\begin{align}
&\eps_\rmv^{(2)} ( f_{\rmv,n},\varphi_n)\nn\\*
 & = \int_\calZ D\big( \varphi_n\cdot f_{\rmv,n}(z)  \big\|P_{X|Z=z}^n \big) \, \mu(\rmd z)\\
& = \int_\calZ D\bigg( \sum_{y\in\calY_n} P_{f_{\rmv,n}(z)}(y)  \varphi_n(y) \bigg\|P_{X|Z=z}^n \bigg) \, \mu(\rmd z)\\
&=\int_\calZ \sum_{y\in\calY_n} P_{f_{\rmv,n}(z)}(y)  \Bigg[ D\big(\varphi_n(y) \big\| P_{X|Z=z}^n\big)   \nn\\*
&\quad -   D\bigg( \varphi_n(y) \bigg\| \sum_{y'\in\calY_n} P_{f_{\rmv,n}(z)}(y')\varphi_n(y')  \bigg)\Bigg]\, \mu(\rmd z) . \label{eqn:apply_pyt2}
\end{align}
We analyze both terms in  \eqref{eqn:apply_pyt2} in turn. 

For the first term, we use an argument similar to that for the proof of Theorem 1(a) in Rissanen~\cite{Rissanen84}. Note that since the size of $|\calY_n|$ is $n^{ \frac{d}{2}(1-\epsilon)}$, the set $\calS_{n,\epsilon}$ of all possible distributions output by the decoder $\varphi_n$ cannot exceed $n^{ \frac{d}{2}(1-\epsilon)}$, i.e., $|\calS_{n,\epsilon}|=|\{ \varphi_n(y):y\in\calY_n\}|\le |\calY_n|= n^{ \frac{d}{2}(1-\epsilon)}$. For any given $z\in\calZ$, let the closest distribution in $\calS_{n,\epsilon}$ have parameter $z' \in\calZ$. Since $\calZ\in\bbR^d$ is  bounded, we can estimate the (order of the) $\ell_2$ distance  between $z$ and $z'$, i.e., $\Delta:=\| z-z'\|$. If $z$ is a point in general position in $\calZ$, then $\Delta$ is of the same order as $r$, where $r$ is the largest radius of the $n^{ \frac{d}{2}(1-\epsilon)}$ disjoint spheres   contained in $\calZ$. Since the volume spheres of radius $r$ in $\bbR^d$ is proportional to $r^d$, we have that 
\begin{equation}
 K_d\cdot r^d \cdot n^{ \frac{d}{2}(1-\epsilon)} \ge \mathrm{vol}(\calZ),
\end{equation}
where $K_d$ is a constant that depends only on the dimension $d$. 
Since $ \mathrm{vol}(\calZ)$ does not depend on $n$ (it also only depends on $d$) and $\Delta=\Theta( r)$,\footnote{The implied constants in the $\Omega(\fndot)$ notations used in \eqref{eqn:Delta}, \eqref{eqn:first_term}, and~\eqref{eqn:eps_to_inf}   are all assumed to be positive.}
\begin{equation}
\Delta=\Omega(n^{-  \frac{1}{2} (1- \epsilon)}). \label{eqn:Delta}
\end{equation}
At the same time, by the Euclidean approximation of relative entropy in \eqref{eqn:euc_re},  $D( P_{X|Z=z'}^n\| P_{X|Z=z }^n) = \Omega( n \| z-z'\|^2)=\Omega( n^{ \epsilon})$. Thus the first term in \eqref{eqn:apply_pyt2} scales as 
\begin{equation}
 \!\!\!\!\!\int_\calZ \sum_{y\in\calY_n}\!\! P_{f_{\rmv,n}(z)}(y) D\big(\varphi_n(y) \big\| P_{X|Z=z}^n\big) \, \mu(\rmd z) =\Omega( n^{ \epsilon}).  \label{eqn:first_term}
\end{equation}

On the other hand the second term in \eqref{eqn:apply_pyt2} is a  conditional mutual information. In particular, it can be upper bounded as
\begin{align}
&  \int_\calZ \sum_{y\in\calY_n}\! P_{f_{\rmv,n}(z)}(y) D\bigg(\!  \varphi_n(y) \bigg\| \! \sum_{y'\in\calY_n} \! \! \! P_{f_{\rmv,n}(z)}(y')\varphi_n(y')  \bigg)\mu(\rmd z)   \nn\\*
  &= I(X^n;Y|Z)\le H(Y)  \le   \frac{d}{2}(1-\epsilon) \log n.\label{eqn:sec_term}
\end{align}
Note that the random variables in the    information quantities above are computed with respect to the distribution induced by the visible code $\calC_{\rmv,n} = (f_{\rmv,n},\varphi_n)$.

Combining  \eqref{eqn:apply_pyt2}, \eqref{eqn:first_term}, and \eqref{eqn:sec_term}, we obtain
\begin{equation}
\eps_\rmv^{(2)} ( f_{\rmv,n},\varphi_n)\ge \Omega(n^{ \epsilon})- \frac{d}{2}(1-\epsilon) \log n\to \infty. \label{eqn:eps_to_inf}
\end{equation}
Hence, with a memory size of $n^{ \frac{d}{2}(1-\epsilon)}$, the error computed according to the relative entropy criterion for any visible code     diverges. This completes the proof of  \eqref{eqn:pyt_res}.
\end{proof}
\subsection{Strong Converse Under The Variational Distance Criterion} \label{sec:str_conv}

In this section, we prove the following strong converse statement with respect to the variational distance error criterion.
\begin{lemma} \label{lem:str}
The following lower bound holds
\begin{equation}
R_{\rmv}^{(1)}(\delta_1)\ge\frac{d}{2},\quad\forall\, \delta_1\in [0,2).  \label{eqn:str_conv_res}
\end{equation}
\end{lemma}
Lemma \ref{lem:str} significantly strengthens Lemma \ref{lem:clarke_barron} because the asymptotic error $\delta_1$ is not restricted to be $0$; rather it can take any value in $[0,2)$. It demonstrates the lower bound to  \eqref{eqn:vis1}.

\begin{proof}[Proof of Lemma \ref{lem:str}]
We first consider the case in which $d=1$. We proceed  by contradiction.
We assume, without loss of generality,  that the parameter space $\calZ=[0,1]$.
Fix $\eta \in (0,1/2)$ and assume that the memory size $M_n= |{\cal Y}_n|$ is $O(n^{\frac{1}{2}-\eta})$
and 
\begin{align}
\eps^{(1)}_{\rmv }(f_{\rmv,n},\varphi_n) &:=\bbE_{z\sim \mu} \left[
\big\| \varphi_n\cdot  f_{\rmv,n}(z)- P_{X|Z=z}^n
\big\|_1 \right]   \nn\\*
 &  \le 2-\alpha
\end{align}
for some $\alpha \in (0,2)$ and $n$ large enough. Let 
\begin{equation}
 \calS:=\left\{ z \in\!\calZ :
\big\| \varphi_n\cdot  f_{\rmv,n}(z)\!-\! P_{X|Z=z}^n
\big\|_1
 \le 2-\frac{\alpha}{2}  \right\}. \!
\end{equation}
Markov's inequality implies that
\begin{align}
\mu(\calS)
&\ge 1- \frac{\bbE_{z\sim \mu} \big[
 \| \varphi_n\cdot  f_{\rmv,n}(z)- P_{X|Z=z}^n
 \|_1 \big]}{2-\frac{\alpha}{2}}\\
& \ge 1- \frac{2-\alpha}{2-\frac{\alpha}{2}} =  \frac{\alpha}{4-\alpha}>0. \label{eqn:muge0}
\end{align}
Let $\lambda$ be the Lebesgue measure on $[0,1]$. From \eqref{eqn:muge0}, we know that $\lambda(\calS)>0$ by absolute  continuity of $\mu$ with respect to $\lambda$ (see Section~\ref{sec:setup}). Thus, we may choose $\frac{5}{\alpha}M_n$ points $\{z_i : i = 1,\ldots,\frac{5}{\alpha} M_n\} \subset\calZ$ satisfying the following two conditions:
\begin{align}
\big\| \varphi_n\cdot  f_{\rmv,n}(z_i)- P_{X|Z=z_i}^n \big\|_1
 & \le 2-\frac{\alpha}{2}\label{H1} , \quad\mbox{and}\\
\forall\, i\ne j , \quad |z_i-z_j| & > { \lambda(\calS)}    \Big(   \frac{5}{\alpha}M_n   \Big)^{-1},
\end{align}
Since $ {\lambda(\calS)  }  \big(   \frac{5}{\alpha}M_n   \big)^{-1}=\Omega(n^{-\frac{1}{2}+\eta})$,
the distributions in the set $\{P_{X|Z=z_i}^n :  i = 1,\ldots,\frac{5}{\alpha} M_n\} \subset\calP(\calX^n)$  are {\em distinguishable}. That is,
  for any $\epsilon  >0$ we may choose an $N \in\bbN$ satisfying the following.
For any $n \ge N$, there exists disjoint subsets $\calD_i \subset {\cal X}^n$
such that
\begin{align}
P_{X |Z=z_i}^n(\calD_i) \ge 1-\epsilon \label{H2}
\end{align}
for any $i=1, \ldots, \frac{5}{\alpha}M_n$. For example, we may take
\begin{equation}
\calD_i \!:=\! \bigg\{\! x^n \!\in\!\calX^n :\! \Big| \frac{1}{n}\sum_{j=1}^n x_j - z_i  \Big| \!\le\! \frac{\lambda(\calS)}{3} \cdot \Big(   \frac{5}{\alpha}M_n   \Big)^{-1} \bigg\}
\end{equation}
and it is then easy to verify that $\{\calD_i :i=1,\ldots, \frac{5}{\alpha}M_n\}$ are disjoint and, by Chebyshev's inequality, that  \eqref{H2} holds for $n$ large enough.
Now recall that for any two probability measures $P,Q$ on a common probability  space $(\Omega,\scF)$, half the variational distance can be expressed as $\frac{1}{2} \| P-Q\|_1=\sup\{ P(\calA)-Q(\calA) : \calA\in\scF\}$. 
Thus, the combination of~\eqref{H1} and~\eqref{H2} shows that
\begin{align}
 1-\frac{\alpha}{4}&\ge \big(\varphi_n\cdot  f_{\rmv,n}(z_i)\big)(\calD_i^c)- P_{X|Z=z_i}^n(\calD_i^c)\\
 &\ge \big(\varphi_n\cdot  f_{\rmv,n}(z_i)\big)(\calD_i^c)-\epsilon.
\end{align}
In other words,
\begin{align}
\big( \varphi_n\cdot  f_{\rmv,n}(z_i)\big)
(\calD_i) \ge \frac{\alpha}{4}-\epsilon \label{H3}.
\end{align}
We denote the elements of ${\cal Y}_n$ by $\{1, \ldots, M_n\}$.
The distribution  at the output of the decoder $\varphi_n\cdot  f_{\rmv,n}(z)$ is a convex combination of
$\{\varphi_n(1), \ldots, \varphi_n(M_n)\}$. Thus, 
\begin{align}
\sum_{j=1}^{M_n}\big(\varphi_n(j)\big)(\calD_i)
\ge\big( \varphi_n\cdot  f_{\rmv,n}(z_i)\big)(\calD_i) . \label{eqn:cvx_comb}
\end{align}
Hence,
\begin{align}
M_n 
&\ge \sum_{j=1}^{M_n}\big(\varphi_n(j) \big)\bigg(
 \bigcup_{i=1}^{\frac{5}{\alpha}M_n} \calD_i\bigg)\\
&= \sum_{i=1}^{\frac{5}{\alpha}M_n} \sum_{j=1}^{M_n}\big(\varphi_n(j)\big)(\calD_i)\\
&\ge \sum_{i=1}^{\frac{5}{\alpha}M_n}\big(  \varphi_n\cdot  f_{\rmv,n}(z_i)\big)(\calD_i)\label{eqn:H4}\\
&\ge \sum_{i=1}^{\frac{5}{\alpha}M_n} \Big(\frac{\alpha}{4}-\epsilon \Big)= \frac{5}{\alpha}M_n  \Big(\frac{\alpha}{4}-\epsilon\Big),\label{eqn:H5} 
\end{align}
where \eqref{eqn:H4}  and the inequality in~\eqref{eqn:H5} are applications of   the bounds in \eqref{eqn:cvx_comb} and  \eqref{H3} respectively. 
So, we obtain
\begin{align}
1\ge\frac{5}{\alpha} \Big(  \frac{\alpha}{4}-\epsilon\Big),
\end{align}
which is a contradiction (if $\epsilon>0$ is chosen to be smaller than $\frac{ \alpha}{20}$). Hence, a memory size of $|\calY_n|=O(n^{\frac{1}{2}-\eta})$ is insufficient to ensure that $\varlimsup_{n\to\infty}\eps_{\rmv}(f_{\rmv,n}, \varphi_n)$  is strictly smaller than $2$.

In the general case in which we assume for the sake of contradiction that when the memory size is $|\calY_n| = O(n^{ d (\frac{1 }{2}-\eta)})$ (for fixed $\eta>0$), per dimension, the memory size is of the order  $O(n^{\frac{1}{2} -\eta})$. Now, we can treat each dimension separately and apply the above argument to yield the same contradiction when the memory size is too small.
\end{proof}

\section{Discussion and Future Research Directions} \label{sec:discuss}
In this paper, we have considered the approximate reconstruction of a generating distribution $P_{X|Z=z}^n$ from a compressed version of a source $X^n$ (the blind setting) or the parameter of generating distribution $z$   itself (the visible setting). We have shown using various notions of approximate sufficient statistics that  when $n$  i.i.d.\ observations $X^n$ are available, the length of the optimal code in most cases and under suitable regularity conditions  is $\frac{d}{2}\log n + o(\log n) $. In the process of deriving our results, we have strengthened the achievability part based on Rissanen's MDL principle \cite{Rissanen83, Rissanen84}. We have also proved strong converses, thus strengthening the utility of Clarke and Barron's formula~\cite{clarke90, clarke94}, which, by itself, only leads to a weak converse under the variational distance error criterion. 

There are three promising avenues for future research:
\begin{enumerate}
\item It is natural to question whether the assumption of $\{ P_{X|Z=z} \}_{z\in\calZ}$ being an {\em exponential family} is necessary to achieve $R_{\rmb}^{ (2) }(\delta_2')   = \frac{d}{2},$  for all $\delta_2' \in [0,\infty)$ in  \eqref{eqn:blind2_exp} (i.e., asymptotically zero error). We would like to remove this somewhat restrictive assumption if possible. 
\item It is also natural to wonder about the scaling and form of the {\em second-order term} in the optimal code length   $\log |\calY_n|$. It is known from Theorem \ref{thm:vis} that, in most cases, the first-order term scales as $\frac{d}{2}\log n$. Our achievability results  based on Lemmas \ref{lem:mdl}--\ref{lem:gen}  suggest that the second-order term is of the constant order $O(1)$. Establishing that this is indeed a constant and the dependence of this constant as a function of $\delta \ge 0$, the bound on the error, would be of significant theoretical and practical interest. 
\item Since there are many distance ``metrics'' that may be used for measuring   distances between two distributions (e.g., Csisz\'ar's $f$-divergences) \cite{Liese06}, it may also be fruitful to study the asymptotics of the optimal code length  $\log |\calY_n|$ when other distance measures beyond the relative entropy and variational distance are used to quantify the discrepancy between $P_{X|Z=z}^n$ and $\varphi_n\cdot f_{\rmv,n}\cdot P_{X|Z=z}^n$ (in the blind setting).
\end{enumerate}
\appendices

\subsection*{Acknowledgements}   
The authors thank the Associate Editor Dr.\ Peter Harremo\"es and the two reviewers for their insightful comments that helped to  improve the paper significantly.

\bibliographystyle{unsrt}
\bibliography{isitbib}

\begin{IEEEbiographynophoto}{Masahito Hayashi}(M'06--SM'13--F'17) was born in Japan in 1971.
He received the B.S.\ degree from the Faculty of Sciences in Kyoto
University, Japan, in 1994 and the M.S.\ and Ph.D.\ degrees in Mathematics from
Kyoto University, Japan, in 1996 and 1999, respectively. He worked in Kyoto University as a Research Fellow of the Japan Society of the
Promotion of Science (JSPS) from 1998 to 2000,
and worked in the Laboratory for Mathematical Neuroscience,
Brain Science Institute, RIKEN from 2000 to 2003,
and worked in ERATO Quantum Computation and Information Project,
Japan Science and Technology Agency (JST) as the Research Head from 2000 to 2006.
He also worked in the Superrobust Computation Project Information Science and Technology Strategic Core (21st Century COE by MEXT) Graduate School of Information Science and Technology, The University of Tokyo as Adjunct Associate Professor from 2004 to 2007.
He worked in the Graduate School of Information Sciences, Tohoku University as Associate Professor from 2007 to 2012.
In 2012, he joined the Graduate School of Mathematics, Nagoya University as Professor.
He also worked in Centre for Quantum Technologies, National University of Singapore as Visiting Research Associate Professor from 2009 to 2012
and as Visiting Research Professor from 2012 to now.
He also worked in Center for Advanced Intelligence Project, RIKEN as
a Visiting Scientist from 2017.
In 2011, he received Information Theory Society Paper Award (2011) for ``Information-Spectrum Approach to Second-Order Coding Rate in Channel Coding''.
In 2016, he received the Japan Academy Medal from the Japan Academy
and the JSPS Prize from Japan Society for the Promotion of Science.

In 2006, he published the book ``Quantum Information: An Introduction''  from Springer, whose revised version was published as ``Quantum Information Theory: Mathematical Foundation'' from Graduate Texts in Physics, Springer in 2016.
In 2016, he published other two books ``Group Representation for Quantum Theory'' and ``A Group Theoretic Approach to Quantum Information'' from Springer.
He is on the Editorial Board of {\it International Journal of Quantum Information}
and {\it International Journal On Advances in Security}.
His research interests include classical and quantum information theory and classical and quantum statistical inference.
\end{IEEEbiographynophoto}
\begin{IEEEbiographynophoto}{Vincent Y.\ F.\ Tan} (S'07-M'11-SM'15) was born in Singapore in 1981. He is currently an Assistant Professor in the Department of Electrical and Computer Engineering (ECE) and the Department of Mathematics at the National University of Singapore (NUS). He received the B.A.\ and M.Eng.\ degrees in Electrical and Information Sciences from Cambridge University in 2005 and the Ph.D.\ degree in Electrical Engineering and Computer Science (EECS) from the Massachusetts Institute of Technology in 2011. He was a postdoctoral researcher in the Department of ECE at the University of Wisconsin-Madison and a research scientist at the Institute for Infocomm (I$^2$R) Research, A*STAR, Singapore. His research interests include network information theory, machine learning, and statistical signal processing.

Dr.\ Tan received the MIT EECS Jin-Au Kong outstanding doctoral thesis prize in 2011, the NUS Young Investigator Award in 2014, the  Engineering Young Researcher Award  in the Faculty of Engineering, NUS in 2018,  and the Singapore National Research Foundation (NRF) Fellowship (Class of 2018). He has authored a research monograph on {\em ``Asymptotic Estimates in Information Theory with Non-Vanishing Error Probabilities''} in the Foundations and Trends in Communications and Information Theory Series (NOW Publishers). He is currently an Editor of the IEEE Transactions on Communications and the IEEE Transactions on Green Communications and Networking.
\end{IEEEbiographynophoto}

\end{document}